\documentclass[10pt,a4paper,twocolumn]{article}

\usepackage[utf8]{inputenc}
\usepackage[T1]{fontenc}
\usepackage{amsmath,amssymb,amsthm,mathtools}
\usepackage{bm}
\usepackage{mathrsfs}
\usepackage{graphicx}
\usepackage{geometry}
\usepackage[hidelinks]{hyperref}
\usepackage[super,sort&compress]{natbib}
\usepackage[nameinlink,noabbrev]{cleveref}
\usepackage[ruled,vlined,linesnumbered]{algorithm2e}
\usepackage{authblk}
\usepackage{microtype}
\usepackage{braket}
\usepackage[table]{xcolor}
\usepackage{chngcntr}

\newtheorem{theorem}{Theorem}
\newtheorem*{informaltheorem}{Theorem (Informal)}
\newtheorem{proposition}[theorem]{Proposition}

\newtheorem{corollary}[theorem]{Corollary}

\newcommand{\Data}{\mathcal{D}}

\newcommand{\Loss}{\mathcal{L}}

\newcommand{\poly}{\mathrm{poly}}
\newcommand{\polylog}{\mathrm{polylog}}

\DontPrintSemicolon
\title{Efficient Quantum Algorithm for Robust Training}

\author[1,2]{Yue Wang}
\author[3]{Guangyi He}
\author[5]{Liepeng Zhang}
\author[4,3]{Lukas Gonon}
\author[1]{Qi Zhao}

\affil[1]{QICI Quantum Information and Computation Initiative, School of Computing and Data Science, The University of Hong Kong, Pokfulam Road, Hong Kong SAR, China}
\affil[2]{Shenzhen International Quantum Research Institute, Shenzhen 518000, Guangdong, China}
\affil[3]{Department of Mathematics, Imperial College London}
\affil[4]{School of Computer Science, University of St. Gallen}
\affil[5]{Department of Physics and Astronomy, University College London}

\date{}

\begin{document}

\maketitle

\begin{abstract}
Adversarial training is a standard defense against malicious input perturbations in security-critical machine-learning systems. Its main burden is structural: before every parameter update, the current model must first be attacked to find a new adversarial perturbation, making training increasingly expensive and hard to sustain at large-model scale. Here we give an end-to-end quantum procedure for projected-gradient robust training under local stability and sparsity assumptions. The key step is to reformulate the coupled attacker--learner dynamics as a high-dimensional sparse linear system whose terminal block yields the final network-parameter state. In this formulation, the dominant query cost scales linearly with training time steps, up to logarithmic factors, and polylogarithmically with model size, while the full gate complexity records separate input-preparation and sparse-access overheads. This places core computational tasks for AI security on a concrete quantum footing and identifies a regime in which robust-training overhead can be reduced.

\end{abstract}

\section*{Introduction}

\begin{figure*}[!t]
\centering
\includegraphics[width=0.88\textwidth]{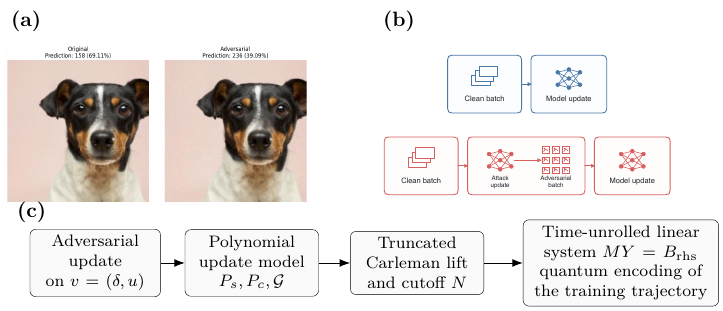}
\caption{\textbf{From robust training to a quantum-accessible linear system.} \textbf{(a)} A visually subtle perturbation changes the prediction of a fixed image classifier. \textbf{(b)} In robust training, each outer update may contain $K_t$ attack substeps and $L_t$ learner substeps. \textbf{(c)} Our reduction approximates one robust-training outer update by a polynomial map, lifts the training window to a high dimensional sparse linear system, and uses its normalized solution as a coherent encoding of the training trajectory.}
\label{fig:pipeline}
\end{figure*}

As machine-learning systems are deployed in applications where incorrect predictions can have serious consequences, securing them against adversarial manipulation has become increasingly important. Such failures now appear across deployed settings: adversarial inputs can evade face-recognition systems\cite{sharif2016accessorize}, mislead visual models in physical environments\cite{eykholt2018robust}, disrupt neural language systems\cite{belinkov2018synthetic,wallace2019universal}, and induce harmful responses from aligned large language models or expose memorized training data\cite{zou2023universal,carlini2021extracting}. Across these settings, small but strategically chosen changes to the input can produce disproportionately large changes in model behaviour while leaving the input visually or semantically close to the original data\cite{szegedyIntriguingPropertiesNeural,goodfellow2015explaining,madry2018towards,zhangTheoreticallyPrincipledTradeoff2019}. As illustrated in Fig.~\ref{fig:pipeline}a, an adversarial example differs from the original image by an imperceptible perturbation, yet a standard ResNet misclassifies it. These examples show that adversarial vulnerability is not confined to one benchmark or modality, but is a general security problem across modern machine-learning systems.

Robust training is a standard way to defend against adversarial manipulation\cite{schmidt2018robustgeneralization,wong2018provable,cohen2019certified,huangRevisitingResidualNetworks2022a}. At a high level, robust training repeatedly updates a model on perturbed inputs chosen to reveal the weaknesses of the current model. As illustrated in Fig.~\ref{fig:pipeline}b, the basic computational unit is a coupled attacker--learner procedure in which the perturbation is computed for the current model and the model is then updated on the attacked input within the same outer iteration\cite{shafahi2019free,wong2020fast,rice2020overfitting,gowalUncoveringLimitsAdversarial2021}. That coupling is also what makes robust training expensive. Ordinary training updates the model once per batch. The extra cost in robust training comes from computing an attack for the current model repeatedly. Over many iterations, this loop makes robust training substantially more expensive than ordinary training. This overhead matters most at scale, where training is already expensive and robust-training pipelines drift away from compute-efficient operating regimes\cite{bartoldsonAdversarialRobustnessLimits2024b}.

This bottleneck makes robust training a natural place to ask what quantum methods might contribute\cite{seekingQuantumAdvantageML2023,aaronson2015fineprint,tang2019quantuminspired,chia2022dequantizing}. A natural first comparison point is quantum machine learning itself. Recent work has argued that quantum learning systems may provide routes to adversarial robustness or provable robustness in specific classification settings\cite{towardsQuantumEnhancedAdversarialRobustness2023,optimalProvableRobustnessQHT2021}. But these works ask whether the learned model can itself be quantum and robust, not whether the projected-gradient attacker loop of classical robust training can be executed through a rigorous quantum reduction. Moreover, large-scale trainability of variational quantum models faces its own obstacle: barren-plateau phenomena can make optimization gradients exponentially small even before robustness enters the picture\cite{mccleanBarrenPlateaus2018,cerezoCostFunctionDependent2021}. Quantum machine learning is therefore a natural comparison point, but not yet a solution to the particular workload bottleneck considered here. Another nearby line of work studies quantum linear-algebra and linearized-dynamics primitives\cite{harrowQuantumAlgorithmSolving2009,childsQuantumAlgorithmSystems2017,gilyen2019qsvt,berry2017linearDE,krovi2023improvedDE}, which is especially efficient when the dominant cost is concentrated in a large, structured high-dimensional computation. These tools become useful only after the target computation has already been cast into a sparse linear system. Robust training does not arrive in that form yet. The attack is parameter-dependent and must be recomputed after each learner update, while sign, clipping and projection remain inside the same coupled outer step. The unresolved point is therefore whether this repeated attacker loop itself can be rewritten as a sparse linear system with a specified output state under explicit validity conditions.

Here we give an end-to-end quantum procedure that recovers a classical description of the final network-parameter state produced by robust training, under local stability and sparsity assumptions. The key step is to rewrite the coupled attacker--learner computation as a sparse linear system whose solution gives a representation of the full training trajectory. Starting from one robust-training outer update on the joint perturbation--parameter state, we replace the local nonsmooth pieces by a polynomial surrogate that preserves the alternating dependence between attack and learning, and then apply discrete-time Carleman lifting\cite{liuProvablyEfficientQuantum2024,forets2017carlemanError,amini2022finiteSectionCarleman} across the training horizon. The resulting quantum subroutine prepares the global trajectory state, and the final parameter is then extracted from its terminal degree-1 block under explicit terminal-weight and sparse-output readout assumptions\cite{cramerEfficientQuantumState2010}.

Our main result is a rigorous mathematical reformulation that maps projected-gradient robust training in a local stable regime to an end-to-end quantum algorithm. The main text develops and discusses the reduction at high level, whereas Methods states the formal procedure and the Appendix contains the detailed surrogate construction, input-model specializations, truncation bounds, error propagation and proofs.

\section*{Results}

The main theorem gives the query complexity of solving the sparse linear system associated with a fixed projected-gradient robust-training window and, under additional readout assumptions, recovering the terminal parameter block.

\begin{informaltheorem}
Consider a fixed $T$-step projected-gradient robust-training window. Suppose the coupled attacker--learner update admits a local polynomial surrogate, the corresponding truncated discrete-time Carleman lift of order $N$ remains stable over the window, and the resulting horizon system is sparse. Then a quantum procedure produces a classical description of the final network-parameter state to output error at most $\varepsilon_{\mathrm{out}}$, with dominant query complexity
$$
\widetilde{\mathcal O}\left(
 s_M\kappa_2(M)
 \mathrm{polylog}\frac{N_h}{\varepsilon_{\mathrm{LS}}}
\right).
$$
\end{informaltheorem}

The theorem replaces the repeated inner-loop attack computation by a single sparse linear system. The quantum linear-system procedure then solves this system only polylogarithmically in the lifted horizon dimension $N_h=(T+1)\Delta_N$, where $\Delta_N=\sum_{j=1}^{N}d^j=O(d^N)$ and $d = m + n$ is the joint perturbation--parameter dimension. The resulting cost depends on the sparsity $s_M$ and condition number $\kappa_2(M)$ of the horizon matrix. In the contractive local regime used here, the horizon matrix obeys $\kappa_2(M)\le \min\{(1+\rho)/(1-\rho),\,2(T+1)\}$, so the condition number is constant-order when $\rho$ stays bounded away from $1$, and more generally is at worst linear in the training window $T$. The quantum solver keeps polylogarithmic dependence on the precision of the linear system $\varepsilon_{\mathrm{LS}}^{-1}$.

The quantum procedure is summarized in the box below.

\begin{center}
\fcolorbox{black}{gray!10}{%
\parbox{\dimexpr\columnwidth-2\fboxsep-2\fboxrule\relax}{%
\small
\textbf{Algorithmic summary}

\medskip
\textbf{1. Couple the state.} Write the joint state as
$v_t:=(\delta_t,u_t)\in\mathbb{R}^{m+n}$,
so one robust-training outer step is one coupled state update.

\textbf{2. Approximate one outer step by a polynomial map.} Replace sign, clipping and the required gradient blocks by local polynomial surrogates, so the full outer step becomes a polynomial update
$v_{t+1}=\Psi_t(v_t)$.

\textbf{3. Lift to linear dynamics.} Form the exact truncated monomial vector
$y^{(N)}(t)=(v_t,v_t^{\otimes 2},\ldots,v_t^{\otimes N})$,
and use it to initialize the order-$N$ truncated recurrence
$\hat y(t+1)=B(t)\hat y(t)+c(t)$ with $\hat y(0)=y^{(N)}(0)$.

\textbf{4. Stack the training window.} Collect the lifted states into
$Y=(\hat y(0),\hat y(1),\ldots,\hat y(T))$,
form the forcing vector
$B_{\mathrm{rhs}}=(\hat y(0),c(0),\ldots,c(T-1))$,
and assemble the block sparse horizon system
$MY=B_{\mathrm{rhs}}$.

\textbf{5. Prepare the input, solve, and extract the terminal parameter block.} Prepare the normalized input state
$|B_{\mathrm{rhs}}\rangle\propto B_{\mathrm{rhs}}$,
use a quantum linear-system routine to prepare
$\ket{Y}\propto\sum_{t,i}Y_{t,i}\ket{t,i}$,
then mark the terminal degree-1 parameter block with one ancilla, isolate it, and classically read out the final parameter state under the stated output assumptions.
}}
\end{center}

To state the full complexity, we additionally introduce the input-preparation term $C_{\mathrm{prep}}(B_{\mathrm{rhs}})$ and the sparse-access query cost $C_{\mathrm{SA}}$. Efficient preparation routines are known for several structured settings, including integrable distributions, black-box amplitude loading and sparse states\cite{grover2002stateprep,sanders2019blackboxStatePrep,ramacciotti2024sparseStatePrep}. Under the qRAM-style specialization stated in the Appendix, one may further take $C_{\mathrm{prep}}(B_{\mathrm{rhs}})=\widetilde{\mathcal O}(\mathrm{polylog}\,N_h)$\cite{giovannetti2008qram,hann2019qram}. The readout step further requires a lower bound on the weight of the terminal parameter block inside the trajectory state and an output-sparsity/tomography model for the final parameter state\cite{cramerEfficientQuantumState2010}. The formal theorem is stated in Methods, and the proof details are deferred to the Appendix.

\subsection*{Robust training as a coupled dynamical system}

The key idea behind robust training is to formulate the training as a min-max optimization problem. For a training set $\mathcal{D}$, clean input $x$, label $y$, perturbation $\delta$ and adversarial input $w=x+\delta$, the robust objective is
\begin{equation}
\label{eq:rt_objective}
\min_{u}\frac{1}{|\mathcal{D}|}\sum_{(x,y)\in\mathcal{D}}
\max_{\|\delta\|_p\le \epsilon}
\mathcal{L}\big(f_u(x+\delta),y\big).
\end{equation}
We study a fixed training window with deterministic sample order and step sizes, so the training workload is represented by a concrete sequence of attacker and learner updates. The attacker approximately solves the inner maximization problem over $\delta$ using projected gradient ascent (PGD)\cite{madry2018towards}, where each gradient step is followed by projection back onto the $\epsilon$ ball. The learner then updates the network parameter $u$ by standard gradient descent on the resulting adversarial example.
\begin{equation}
\label{eq:rt_step}
\begin{aligned}
\delta_{t+1}
&=
\Pi_{\mathbb{B}_{\infty}(0,\epsilon)}
\left(
\delta_t+\eta_{\delta}
\operatorname{sign}\big(\nabla_{w}\mathcal{L}(f_{u_t}(x+\delta_t),y)\big)
\right),\\
u_{t+1}
&=
u_t-\eta_{u}\nabla_{u}\mathcal{L}\big(f_{u_t}(x+\delta_{t+1}),y\big).
\end{aligned}
\end{equation}
More general deterministic inner-loop schedules, with $K_t$ attack substeps followed by $L_t$ learner substeps in each outer iteration, are treated as finite compositions of this basic attacker--learner block and are deferred to the Appendix.
The natural state variable is therefore the pair $(\delta_t,u_t)$. This coupled-state formulation treats the attacker and learner as a single discrete-time dynamical system and makes the fixed-window trajectory, rather than only the parameter iterate, the primary state variable in the analysis.

Equation~(\ref{eq:rt_step}) shows why robust training is harder to linearize than ordinary gradient descent. The perturbation cannot be treated as a static input to the learner: it must be recomputed for the current model and then fed back into the same outer step. The perturbation search and learner refinement are therefore interleaved inside one update, and the perturbation step contains sign and projection operations that do not expose a direct linear-algebra form. We therefore work on the enlarged state $v=(\delta,u)$ and reformulate one full outer iteration, in perturbation coordinates, as a local polynomial surrogate that preserves this dependence. Concretely, on a fixed local domain, we replace the sign, clipping and gradient blocks by polynomial approximations, which yields a time-indexed polynomial update map. Discrete-time Carleman lifting then applies exactly to this polynomial surrogate, while the discrepancy from the original robust-training dynamics is controlled separately through the local approximation and truncation analyses.

\subsection*{Folding nonsmooth updates into a polynomial surrogate}

Let $P_{\mathrm{s}}$ be a polynomial approximation to the sign map and $P_{\mathrm{c}}$ a polynomial approximation to coordinatewise clipping. Let $\mathcal{G}_{\delta,t}$ and $\mathcal{G}_{u,t}$ denote polynomial approximations to the perturbation and parameter gradients. In the same simplest $K=L=1$ case, one folded robust-training outer update can be written as
\begin{equation}
\label{eq:folded_step}
\begin{aligned}
\delta^{+}
&=
\epsilon P_{\mathrm{c}}\left(
\frac{\delta+\eta_{\delta,t}P_{\mathrm{s}}(\mathcal{G}_{\delta,t}(v)/\alpha_t)}{\epsilon}
\right),\\
u^{+}
&=
u-\eta_{u,t}\mathcal{G}_{u,t}(\delta^{+},u), v=(\delta,u),
\end{aligned}
\end{equation}
where $\alpha_t$ normalizes the perturbation-gradient approximation on the fixed local domain. This replacement turns one complete robust-training outer step into a polynomial map on the joint state. It therefore preserves the order of attack and learning inside the step while changing the mathematical form of the update into one that can be lifted. The more general deterministic schedule with $K_t$ attack substeps and $L_t$ learner substeps is obtained by finite composition of these basic blocks and remains polynomial, so the later Carleman lifting argument applies to the outer iteration after replacing one-step blocks by their finite composition. Accordingly, the Carleman lift is applied directly to the discrete update implemented during training, rather than to a continuous-time approximation.

Truncating the Carleman lift at order $N$ gives the lifted state dimension
$$
\Delta_N=\sum_{j=1}^{N} d^j, d=m+n,
$$
and stacking the truncated lifted vectors over $T$ steps yields the horizon linear system
\begin{equation}
\label{eq:horizon_system}
M Y = B_{\mathrm{rhs}},
\end{equation}
where $Y=(\hat y(0),\hat y(1),\ldots,\hat y(T))$ is the stacked truncated lifted trajectory and $M$ is block lower bidiagonal. Equation~(\ref{eq:horizon_system}) is the horizon system for the full training window. At this point the repeated attacker--learner computation is represented by a single sparse block-bidiagonal linear system over the full window. The corresponding normalized solution state is
\begin{equation}
\label{eq:trajectory_state}
\ket{Y}
=
\frac{1}{\|Y\|_2}
\sum_{t=0}^{T}\sum_{i=1}^{\Delta_N} Y_{t,i}\ket{t,i}.
\end{equation}
The right-hand side of the horizon system is likewise explicit:
\(
B_{\mathrm{rhs}}=(\hat y(0),c(0),\ldots,c(T-1))^\top
\).
The quantum procedure therefore starts from the normalized input state
\(
\ket{B_{\mathrm{rhs}}}=B_{\mathrm{rhs}}/\|B_{\mathrm{rhs}}\|_2
\),
and we keep its one-time preparation cost explicit as \(C_{\mathrm{prep}}(B_{\mathrm{rhs}})\). For structured inputs this cost can be favorable, with known routines for integrable distributions, black-box amplitude loading and sparse states\cite{grover2002stateprep,sanders2019blackboxStatePrep,ramacciotti2024sparseStatePrep}; under a qRAM-style loading model it can be specialized to polylogarithmic dependence on \(N_h\)\cite{giovannetti2008qram,hann2019qram}.
The quantum linear-system routine prepares this trajectory state. The solver naturally returns the solution of the whole windowed problem, so the final network-parameter state appears inside that global representation as its terminal degree-1 parameter block. Under a constant lower bound on the weight of that block, one ancilla can mark and isolate it with constant success probability; under the output-sparsity assumption, standard state-tomography tools then recover a classical description of the final state\cite{cramerEfficientQuantumState2010}.

\subsection*{Numerical simulation}

We use a reduced MNIST task on digits $0$--$4$ to validate the training process. Each sample is represented by $144$ input features, and the classifier has $60$ trainable parameters. We compare clean-only, robust-only and mixed training with $\mathcal{L}_{\mathrm{mix}}=(1-\alpha)\mathcal{L}_{\mathrm{clean}}+\alpha\mathcal{L}_{\mathrm{rob}}$ at $\alpha=0.50$, so clean-only and robust-only correspond to $\alpha=0$ and $\alpha=1$. We implemented the update matrix and fed the linear system into a HHL QSLP backend\cite{harrowQuantumAlgorithmSolving2009}, which uses the MindSpore Quantum framework~\cite{xu2024mindspore}. Figure~\ref{fig:qrt_numerics} tracks robust accuracy, clean accuracy and clean loss across the three training modes. In all three cases the dynamics show a short transient followed by stable plateaus over the full training window.

\begin{figure*}[t]
\centering
\includegraphics[width=0.84\textwidth]{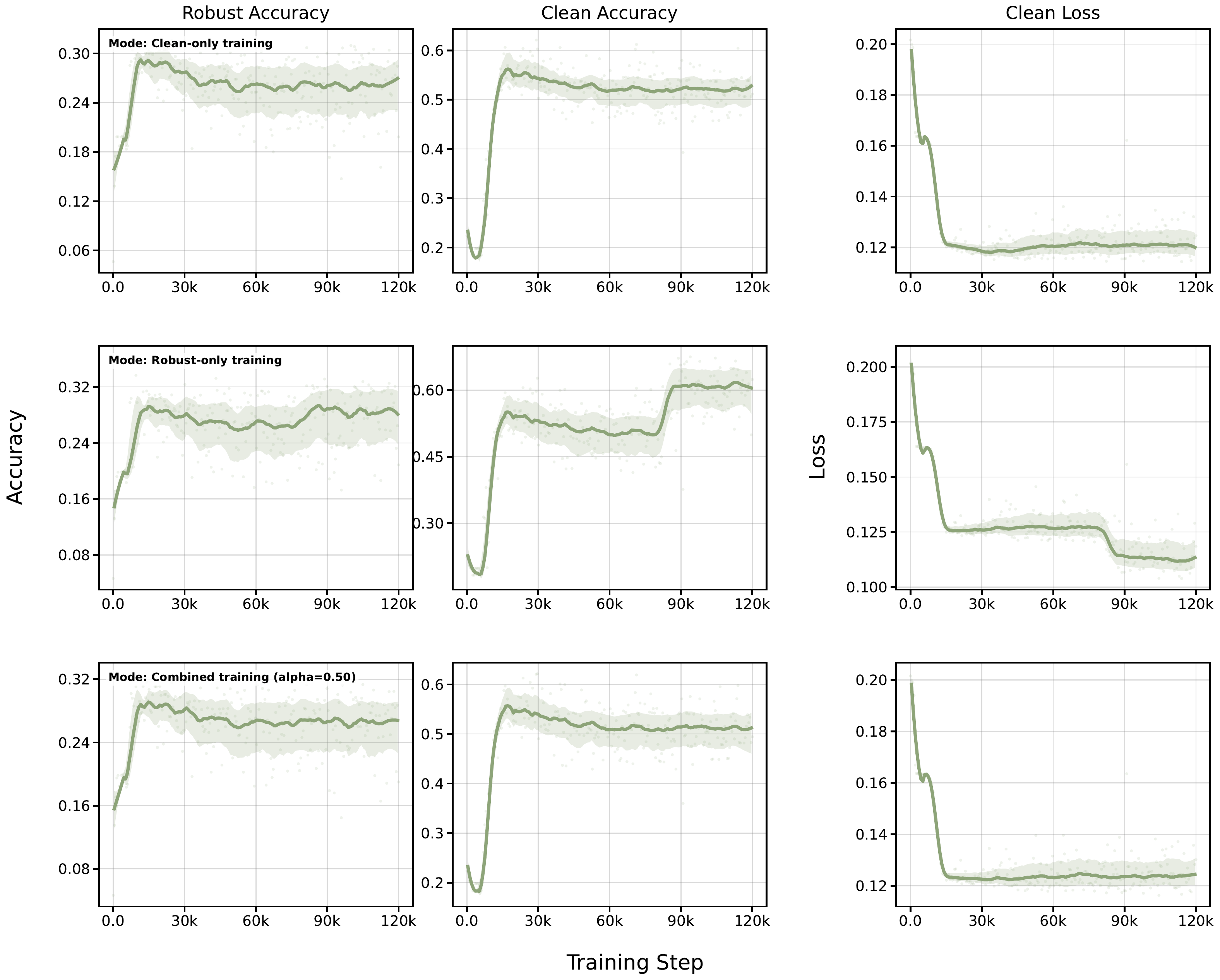}
\caption{\textbf{Numerical verification of the reduced robust-training algorithm on MNIST digits $0$--$4$.} Rows show clean-only, robust-only and mixed training ($\alpha=0.50$); columns show robust accuracy, clean accuracy and clean loss. Robust accuracy is evaluated with a $10$-step PGD attack. Each trajectory shows a short transient followed by stable plateaus over $1.2\times 10^5$ steps. The figure illustrates the behavior of the reduced update model used in the analysis.}
\label{fig:qrt_numerics}
\end{figure*}

\section*{Discussion}

We have shown that, under local stability and sparsity assumptions, a quantum procedure can return a classical description of the final network-parameter state produced by robust training. The core reduction rewrites the coupled attacker--learner computation over a finite training segment as a sparse linear system, so the repeated attacker loop is no longer treated as a sequence of separate optimization calls but as one structured linear-algebra object over the whole window.

Three technical directions remain. First, sharpen the transfer bound from the polynomial surrogate to the exact nonsmooth dynamics with additional practical structured assumptions, so that the approximation and truncation budgets remain controlled over longer training windows. Second, identify terminal quantities that are both security-relevant and measurement-efficient, so that useful outputs can be extracted without full classical recovery of the final parameter state. Third, understand which deterministic inner-loop schedules, model classes and threat models preserve the sparse and well-conditioned structure on which the present reduction relies. Extending the reduction to broader robust-training architectures and to transport-based threat models remains open\cite{sinha2018certifying,mohajerin2018wasserstein}.

\section*{Methods}

\subsection*{Exact robust-training step on the coupled state}
For a fixed sample order $\{(x_t,y_t)\}_{t=0}^{T-1}$ and fixed step sizes $\{(\eta_{\delta,t},\eta_{u,t})\}_{t=0}^{T-1}$, we treat one PGD-style robust-training iteration as a map on the coupled state
$$
v_t:=(\delta_t,u_t)\in\mathbb{R}^{d}, d=m+n.
$$
The perturbation gradient, the pre-projection perturbation, the projected adversarial perturbation and the parameter gradient are
$$
g_{\delta,t}(v_t):=\nabla_w\mathcal{L}\big(f_{u_t}(x_t+\delta_t),y_t\big),
$$
$$
\begin{aligned}
\bar\delta_{t+1}
&=\delta_t+\eta_{\delta,t}\operatorname{sign}\big(g_{\delta,t}(v_t)\big),\\
\delta_{t+1}
&=\Pi_{\mathbb{B}_{\infty}(0,\epsilon)}(\bar\delta_{t+1}),
\end{aligned}
$$
$$
g_{u,t}(\delta_{t+1},u_t)
:=\nabla_u\mathcal{L}\big(f_{u_t}(x_t+\delta_{t+1}),y_t\big),
$$
$$
u_{t+1}=u_t-\eta_{u,t}g_{u,t}(\delta_{t+1},u_t).
$$
Thus one exact iteration produces the sequence
$$
\begin{aligned}
v_t
&\longmapsto \bar\delta_{t+1}\longmapsto \delta_{t+1},\\
&\longmapsto u_{t+1}\longmapsto v_{t+1}:=(\delta_{t+1},u_{t+1}).
\end{aligned}
$$
We later approximate, lift, and stack this coupled trajectory over the full training window.

\subsection*{Polynomial surrogate}

The key obstacle to direct Carleman lifting is that the sign map and the projection are nonsmooth and non-polynomial. We therefore replace sign by a polynomial $P_{\mathrm{s}}$, coordinatewise clipping by a polynomial $P_{\mathrm{c}}$, and the two gradient blocks by polynomial approximations $\mathcal{G}_{\delta,t}$ and $\mathcal{G}_{u,t}$ that uniformly approximate the exact gradients on a fixed local domain. The resulting perturbation and parameter substeps are
$$
\bar\delta^{\mathrm{poly}}_{t+1}
=\delta_t+\eta_{\delta,t}P_{\mathrm{s}}\left(\mathcal{G}_{\delta,t}(v_t)/\alpha_t\right),
$$
$$
\delta^{\mathrm{poly}}_{t+1}
=\epsilon P_{\mathrm{c}}\left(\bar\delta^{\mathrm{poly}}_{t+1}/\epsilon\right),
$$
$$
u^{\mathrm{poly}}_{t+1}
=u_t-\eta_{u,t}\mathcal{G}_{u,t}(\delta^{\mathrm{poly}}_{t+1},u_t),
$$
where $\alpha_t$ rescales the perturbation-gradient surrogate on that region. The full one-step update is therefore a polynomial map
$$
\begin{aligned}
v^{\mathrm{poly}}_{t+1}=\Psi_t(v_t)
&=\sum_{\ell=0}^{D}Q_{\ell}^{(t)}v_t^{\otimes \ell},\\
Q_{\ell}^{(t)}
&\in\mathbb{R}^{d\times d^{\ell}}.
\end{aligned}
$$
The effective degree $D$ is determined by the base polynomial degree $q$ together with the sign and clipping degrees $K_s$ and $K_c$. The approximation degrees and sup-norm error budgets are specified in the Appendix. The main theorem applies to this polynomial surrogate, not to the original nonsmooth dynamics.

\subsection*{Truncated Carleman lift}

For the state sequence generated by the polynomial surrogate, define the lifted monomials
$$
\begin{aligned}
y_0(t)&:=1,\\
y_j(t)&:=v_t^{\otimes j}, j\ge 1,\\
y^{(N)}(t)&:=\bigl(y_1(t),y_2(t),\ldots,y_N(t)\bigr)\in\mathbb{R}^{\Delta_N}.
\end{aligned}
$$
with truncated lifted dimension
$$
\Delta_N=\sum_{j=1}^{N}d^j.
$$
For each lifted level $j$ and input degree $s$, the discrete-time Carleman block is
$$
\begin{aligned}
K_{j,s}(t)
&:=
\sum_{\substack{\alpha\in\{0,1,\ldots,D\}^{j}\\|\alpha|=s}}
Q_{\alpha_1}^{(t)}\otimes\cdots\otimes Q_{\alpha_j}^{(t)},\\
|\alpha|&:=\alpha_1+\cdots+\alpha_j.
\end{aligned}
$$
These blocks yield the exact lifted recurrence
$$
y_j(t+1)=\sum_{s=0}^{jD}K_{j,s}(t)y_s(t), j\ge 1.
$$
Truncating at order $N$, we keep only the blocks with $1\le j,s\le N$ and define
$$
c_j(t):=K_{j,0}(t), c(t):=\bigl(c_1(t),\ldots,c_N(t)\bigr),
$$
$$
B_{j,s}(t):=K_{j,s}(t), 1\le j,s\le N,
$$
which gives the finite lifted recurrence
$$
\begin{aligned}
\hat y(t+1)&=B(t)\hat y(t)+c(t),\\
\hat y(0)&=y^{(N)}(0).
\end{aligned}
$$

\subsection*{Horizon linear system}

The truncated lifted trajectory is stacked over the full training window as
$$
Y:=\bigl(\hat y(0),\hat y(1),\ldots,\hat y(T)\bigr)^{\top},
$$
$$
B_{\mathrm{rhs}}:=\bigl(\hat y(0),c(0),\ldots,c(T-1)\bigr)^{\top}.
$$
The one-step recurrence is equivalent to the block lower-bidiagonal system
$$
MY=B_{\mathrm{rhs}},
$$
$$
M=
\begin{pmatrix}
I & 0 & \cdots & 0\\
-B(0) & I & \ddots & \vdots\\
0 & -B(1) & \ddots & 0\\
\vdots & \ddots & \ddots & I\\
0 & \cdots & -B(T-1) & I
\end{pmatrix}.
$$
This block lower-bidiagonal system is the horizon matrix used in the theorem. When the truncated lifted step matrices satisfy $\sup_t\|B(t)\|_2\le \rho<1$, the horizon matrix is invertible and obeys the condition-number bound
$$
\kappa_2(M)\le \frac{1+\rho}{1-\rho}.
$$
The target of the quantum linear-system routine is the normalized stacked state of equation~(\ref{eq:trajectory_state}), that is, the quantum state proportional to the stacked trajectory vector in Theorem~\ref{thm:main}.

\subsection*{Formal theorem statement}

\begin{theorem}[Quantum recovery of the final parameter state over a fixed training window]
\label{thm:main}
Fix a final-output accuracy budget $\varepsilon_{\mathrm{out}}>0$. For the horizon system $M Y = B_{\mathrm{rhs}}$ associated with the truncated lifted trajectory of the polynomial surrogate dynamics, choose the cutoff $N$ and internal solver tolerance $\varepsilon_{\mathrm{LS}}$ according to the local assumptions described below. Assume further that the normalized right-hand-side state $\ket{B_{\mathrm{rhs}}}$ can be prepared once with cost $C_{\mathrm{prep}}(B_{\mathrm{rhs}})$, that the horizon matrix is available through the stated sparse-access model with query cost $C_{\mathrm{SA}}$, that the terminal degree-1 parameter block of $\ket{Y_N^{\mathrm{ex}}}$ carries probability weight at least $p_*>0$, and that the corresponding final parameter state satisfies the sparse-output tomography model described below. Then a quantum procedure produces, with success probability at least $2/3$, a classical description of the final network-parameter state with total error at most $\varepsilon_{\mathrm{out}}$. Its dominant quantum subroutine prepares a state $\ket{\widetilde Y}$ such that
$$
\bigl\|\ket{\widetilde Y}-\ket{Y_N^{\mathrm{ex}}}\bigr\|_2
\le \varepsilon_{\mathrm{out}},
$$
with query complexity
$$
\widetilde{\mathcal O}\left(
 s_M\kappa_2(M)
 \mathrm{polylog}\frac{N_h}{\varepsilon_{\mathrm{LS}}}
\right),
$$
and gate complexity
$$
\widetilde{\mathcal O}\left(
C_{\mathrm{prep}}(B_{\mathrm{rhs}})
+
s_M\kappa_2(M)C_{\mathrm{SA}}
\mathrm{polylog}\frac{N_h}{\varepsilon_{\mathrm{LS}}}
\right),
$$
where $N_h=(T+1)\Delta_N$ and $\Delta_N=\sum_{j=1}^{N}d^j$.
\end{theorem}

Here $d=m+n$ is the dimension of the coupled perturbation--parameter state, $s_M$ is the sparsity of the horizon matrix $M$, and $\kappa_2(M)$ is its condition number. The parameter $\varepsilon_{\mathrm{LS}}$ is the internal solver tolerance, while $C_{\mathrm{prep}}(B_{\mathrm{rhs}})$ and $C_{\mathrm{SA}}$ denote the costs of preparing the right-hand side and querying the sparse-access matrix oracle, respectively. The additional readout assumptions are a lower bound $p_*$ on the weight of the terminal degree-1 parameter block inside the normalized trajectory state and a sparse-output tomography model for the final parameter state.

\subsection*{Approximation layers and oracle model}

The analysis separates four approximation layers. The first is the modeling error between the exact nonsmooth robust-training outer update and the polynomial surrogate. The second is the Carleman-truncation error from omitting lifted levels above $N$. The third is the solver error from preparing the quantum state associated with the horizon linear system. The fourth is the extraction-and-readout error incurred when the terminal parameter block is isolated and converted into a classical description. The theorem controls the total output error by choosing the polynomial degree parameters, the cutoff $N$, the internal solver tolerance $\varepsilon_{\mathrm{LS}}$ and the readout accuracy budget so that the final classical error remains below $\varepsilon_{\mathrm{out}}$.

The theorem further assumes oracle access to the rescaled system matrix together with a preparation routine for the right-hand side, a constant lower bound on the terminal parameter-block weight, and a sparse-output tomography model for the final state\cite{cramerEfficientQuantumState2010}. This distinction matters. Once the truncated blocks are explicit, a classical baseline can still generate the lifted iterates by forward recursion, whereas the quantum route here first prepares the intermediate trajectory state, then isolates the terminal parameter block with one ancilla, and finally returns a classical description of the final network parameters through readout. The detailed degree--error bounds for the sign approximation,\cite{eremenko2007signApprox} together with the underlying Jackson/Chebyshev approximation background,\cite{devore1968jackson,trefethen2019atap} the cutoff design conditions for the Carleman truncation,\cite{forets2017carlemanError,amini2022finiteSectionCarleman} and the terminal-readout analysis are deferred to the Appendix.

\section*{Acknowledgements}

Q.Z. acknowledges funding from Quantum Science and Technology-National Science and Technology Major Project 2024ZD0301900, National Natural Science Foundation of China (NSFC) via Project No. 12347104 and No. 12305030, Guangdong Basic and Applied Basic Research Foundation via Project 2023A1515012185, Hong Kong Research Grant Council (RGC) via No. 27300823, N\_HKU718/23, and R6010-23, Guangdong Provincial Quantum Science Strategic Initiative No. GDZX2303007.

This work was sponsored by CPS-Yangtze Delta Region Industrial Innovation Center of Quantum and Information Technology-MindSpore Quantum Open Fund.

\section*{Data availability}

The numerical convergence study uses a reduced MNIST classification task on digits $0$--$4$ derived from the publicly available MNIST dataset.

\section*{Code availability}

The source code is publicly available at \url{https://github.com/RamPrime/Quantum-Robust-Training-Mindquantum}.

\clearpage
\onecolumn
\appendix
\numberwithin{equation}{section}
\counterwithin{theorem}{section}
\counterwithin{figure}{section}
\counterwithin{algocf}{section}
\renewcommand{\thefigure}{\thesection.\arabic{figure}}
\renewcommand{\thealgocf}{\thesection.\arabic{algocf}}
\setcounter{theorem}{0}
\setcounter{figure}{0}
\setcounter{algocf}{0}

\section{Problem setup and notation}

This appendix contains the technical statements and proofs behind the main-text claim that a fixed training window of projected-gradient robust training can be rewritten as a sparse linear-system problem under the stated quantum access assumptions. We fix the polynomial update model and separate three approximation layers: the polynomial model, the Carleman truncation, and the QLSA error.

We study robust training as the min--max problem
\begin{equation}
    \label{eq:RT_Optimization}
    \min_{u}\;\frac{1}{|\Data|}\sum_{(x,y)\in\Data}\;\max_{\|\delta\|_p\le \epsilon}\;
    \Loss\big(f_u(x+\delta),y\big),
\end{equation}
where $x\in\mathbb{R}^{m}$ is the clean input, $\delta\in\mathbb{R}^{m}$ is the adversarial perturbation, $w=x+\delta$ is the adversarial input, $u\in\mathbb{R}^{n}$ is the trainable parameter vector, and
\[
v:=\begin{pmatrix}\delta\\ u\end{pmatrix}\in\mathbb{R}^{d},
\qquad d:=m+n,
\]
is the coupled state. Throughout this appendix, $q$ denotes the base polynomial degree of the gradient approximation, $K_s$ and $K_c$ denote the degrees of the sign and clipping polynomials, $D$ is the effective degree of the polynomial update, $N$ is the Carleman cutoff, $\Delta_N=\sum_{j=1}^{N}d^j$ is the lifted state dimension and $N_h=(T+1)\Delta_N$ is the horizon dimension.

\subsection{Polynomial update model of one robust-training step}
\label{sec:training_enlarged_params}

Let $x\in\mathbb{R}^m$ be the clean input, $w\in\mathbb{R}^m$ the adversarial input, and define the perturbation
$\delta:=w-x$ so that the constraint is $\|\delta\|_\infty\le \epsilon$.
Let $u\in\mathbb{R}^n$ denote trainable network parameters, and define the coupled state
\[
v:=\begin{pmatrix}\delta\\ u\end{pmatrix}\in\mathbb{R}^{d},\qquad d:=m+n.
\]
Using perturbation coordinates keeps the projection centered at the origin:
\begin{equation}
\label{eq:projection_shift_identity_clean}
\Pi_{\mathbb B_\infty(x,\epsilon)}(x+\zeta)
=
x+\Pi_{\mathbb B_\infty(0,\epsilon)}(\zeta)
\qquad\text{for all }\zeta\in\mathbb R^m.
\end{equation}
Hence the projector becomes zero-centered in $\delta$-coordinates. Throughout the core theorem we therefore linearize the robust-training dynamics in $v=(\delta,u)$. To simplify notation, we first state the construction for one attack step followed by one learner step; the deterministic \((K_t,L_t)\) case is given later.

We allow a deterministic time-indexed family of polynomial approximations (of base degree $q$) for the relevant gradients as functions of the coupled state $v$:
\begin{align}
g_{\delta,t}(v) &\approx \mathcal{G}_{\delta,t}(v):=\sum_{\ell=0}^{q} G_{\delta,\ell}^{(t)}\,v^{\otimes \ell}\in\mathbb{R}^{m},
\qquad G_{\delta,\ell}^{(t)}\in\mathbb{R}^{m\times d^\ell}, \label{eq:Gdelta_poly}\\
g_{u,t}(v) &\approx \mathcal{G}_{u,t}(v):=\sum_{\ell=0}^{q} G_{u,\ell}^{(t)}\,v^{\otimes \ell}\in\mathbb{R}^{n},
\qquad G_{u,\ell}^{(t)}\in\mathbb{R}^{n\times d^\ell}. \label{eq:Gu_poly}
\end{align}
The time index $t$ may encode a deterministic data schedule, a step-size schedule, or the time-homogeneous case $G_{\bullet,\ell}^{(t)}\equiv G_{\bullet,\ell}$.

For the polynomial $\ell_\infty$-PGD perturbation step, let $P_{\mathrm{s}}$ be an odd polynomial approximating $\operatorname{sign}$ after normalization and
let $P_{\mathrm{c}}$ be an odd polynomial approximating $\mathrm{sat}(x)=\operatorname{clip}(x,-1,1)$ on the relevant domain.
Define the polynomial perturbation update map
\begin{equation}
\label{eq:delta_folded_map}
\delta^{+}
:=
F_t(v)
:=
\epsilon\;P_{\mathrm{c}}\!\Big(\frac{\delta+\eta_{\delta,t}\,P_{\mathrm{s}}\!\big(\mathcal{G}_{\delta,t}(v)/\alpha_t\big)}{\epsilon}\Big),
\qquad (\text{applied coordinate-wise}),
\end{equation}
where $\alpha_t>0$ bounds $\|\mathcal{G}_{\delta,t}(v)\|_\infty$ on the local domain.
Because the state variable is $\delta$, not $w=x+\delta$, the clipping in \eqref{eq:delta_folded_map} acts around zero in perturbation coordinates.

After the perturbation step, we update the parameter block, matching the standard alternating robust-training convention:
\begin{equation}
\label{eq:u_folded_map}
u^{+}
:=
u-\eta_{u,t}\,\mathcal{G}_{u,t}\!\big(F_t(v),u\big).
\end{equation}
Here \(\mathcal G_{u,t}(F_t(v),u)\) is shorthand for evaluating the coupled-state polynomial \(\mathcal G_{u,t}\) at the vector \(H_t(v)=(F_t(v),u)\).
Equivalently, one complete adversarial-training iteration in the polynomial model is the full-step map
\begin{equation}
\label{eq:full_step_update_clean}
v(t+1)=\Psi_t(v(t)),
\qquad
\Psi_t(v)
:=
\begin{pmatrix}
F_t(v)\\
u-\eta_{u,t}\,\mathcal{G}_{u,t}\!\big(F_t(v),u\big)
\end{pmatrix}.
\end{equation}
Since $F_t$ is polynomial and $\mathcal{G}_{u,t}$ is polynomial in its state-vector argument, each $\Psi_t$ is polynomial on $\mathbb R^d$. We therefore write
\begin{equation}
\label{eq:full_step_vector_fields_clean}
\Psi_t(v)=\sum_{\ell=0}^{D}Q_{\ell}^{(t)}\,v^{\otimes \ell},
\qquad
Q_{\ell}^{(t)}\in \mathbb R^{d\times d^\ell}.
\end{equation}
The effective full-step degree $D$ is quantified in Proposition~\ref{prop:folded_effective_degree_clean} below.
The choice between perturbation coordinates $\delta$ and adversarial-input coordinates $w=x+\delta$ is independent of the choice between evaluating the outer update at the current perturbation $\delta$ or the updated perturbation $\delta^{+}$. Using $\delta$ as the state variable is what moves the projection center to zero through \eqref{eq:projection_shift_identity_clean}. The present theorem uses the alternating convention
\(
u^{+}=u-\eta_{u,t}\,\mathcal{G}_{u,t}(\delta^{+},u)\)
because it matches the usual adversarial-training step exactly at the polynomial-model level. A simultaneous variant using the current perturbation,
\(u^{+}=u-\eta_{u,t}\,\mathcal{G}_{u,t}(\delta,u)\),
is a simpler special case with a smaller effective degree and can be treated by the same exact discrete-time lifting argument.

\begin{algorithm}[t]
\caption{Polynomial update model via an exact discrete-time lift of one adversarial-training step}
\label{alg:folded_qrt}
\KwIn{Clean input $x\in\mathbb{R}^m$, label $y$, init $u_0\in\mathbb{R}^n$, radius $\epsilon$, steps $T$,
step sizes $(\eta_{\delta,t},\eta_{u,t})_{t=0}^{T-1}$, Carleman cutoff $N$,
sign and clip degrees $(K_s,K_c)$ chosen later by Item~(4) of \Cref{prop:sign_clip_design_clean},
polynomials $P_{\mathrm{s}}$ (sign), $P_{\mathrm{c}}$ (sat), normalization $(\alpha_t)_{t=0}^{T-1}$,
and a deterministic data/step-size schedule encoded in the full-step family $\{\Psi_t\}_{t=0}^{T-1}$.}
\KwOut{A quantum state proportional to the stacked trajectory vector $Y$, and under the terminal-weight and sparse-output readout assumptions stated later, a classical description of the final parameter state.}

\BlankLine
Set $\delta(0)\gets 0$ and $v(0)\gets (\delta(0),u_0)\in\mathbb{R}^{m+n}$.\;
Construct the polynomial gradient-approximation tensors $\{G_{\delta,\ell}^{(t)}\}_{\ell=0}^{q}$ and $\{G_{u,\ell}^{(t)}\}_{\ell=0}^{q}$ in \eqref{eq:Gdelta_poly}--\eqref{eq:Gu_poly} for each $t=0,1,\dots,T-1$.\;

For each training step $t$, form the polynomial update map $\Psi_t$ from \eqref{eq:delta_folded_map}--\eqref{eq:full_step_update_clean} and expand it into coefficient tensors $\{Q^{(t)}_\ell\}_{\ell=0}^{D}$ as in \eqref{eq:full_step_vector_fields_clean}.\;

For each training step $t$, build the truncated lifted step matrix $B(t)$ and constant block $c(t)$ from Item~(1) of \Cref{prop:carleman_horizon_cond_clean}.\;

Form the horizon matrix $M$ (block-lower-bidiagonal) for the full training window $t=0,1,\dots,T-1$.\;
Prepare the normalized right-hand-side state $|B_{\mathrm{rhs}}\rangle$ at gate cost $C_{\mathrm{prep}}(B_{\mathrm{rhs}})$.\;

Run one linear-system routine for the rescaled horizon system to accuracy $\varepsilon_{\mathrm{LS}}$, using the solver model specified later in hypothesis (H3), and produce a state proportional to the stacked trajectory vector $Y$.\;

If only a quantum encoding of the training trajectory is needed, stop here.\;
Mark the terminal degree-1 parameter block (the parameter coordinates inside the time-$T$ level-1 block) with one ancilla, postselect or amplitude-amplify that marked sector, and run the sparse-output tomography/readout routine assumed later in hypothesis (H6).\;
\end{algorithm}

\subsection{Main finite-step theorem}
\label{subsec:one_shot_folded_clean_spine}

\begin{figure}[t]
    \centering
    \includegraphics[width=\textwidth]{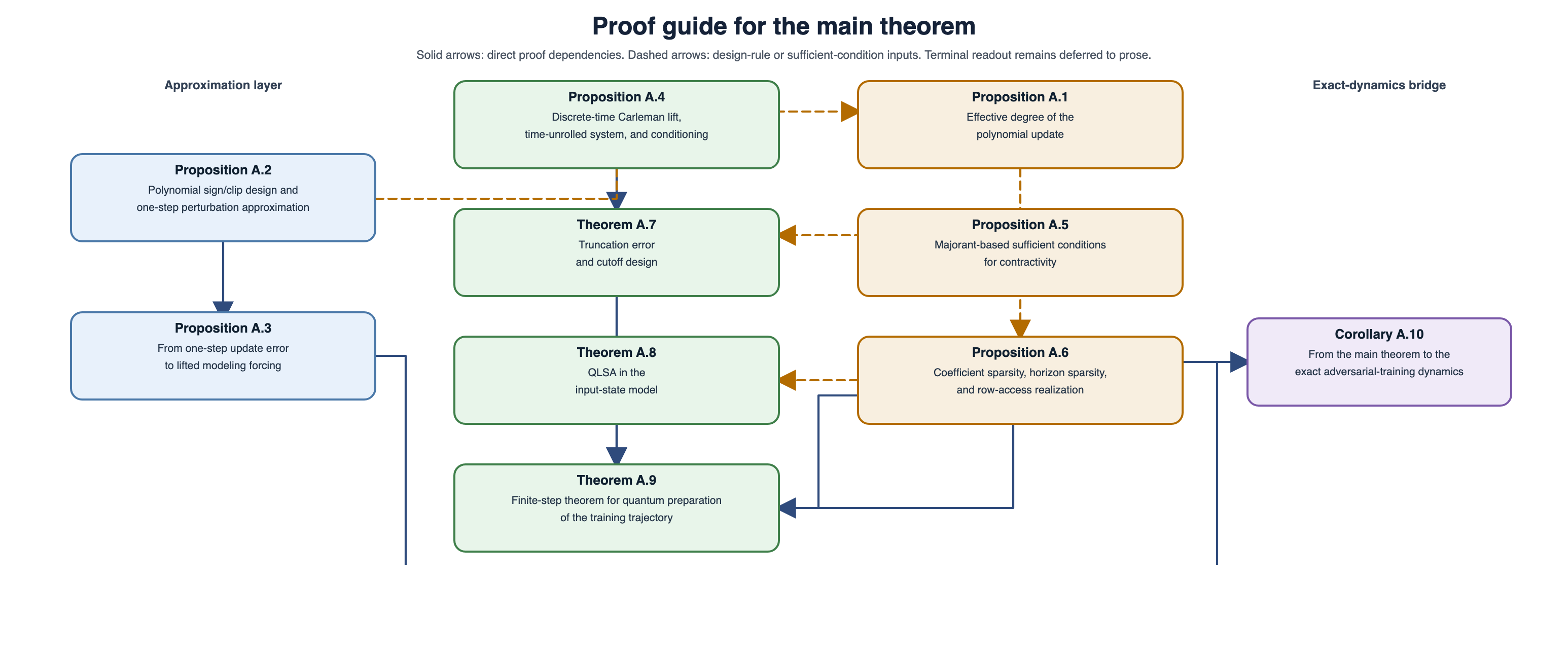}
    \caption{Logical dependencies among the main propositions and theorems.}
    \label{fig:proof_relation_overview}
\end{figure}

This subsection states the main theorem for the polynomial update model. Once the horizon system is explicit, a classical simulation can still generate its lifted trajectory by forward recursion; the theorem establishes conditions under which the same stacked trajectory can be prepared as a quantum state and its terminal degree-1 parameter block can be read out efficiently.
For reference, the exact \(\ell_\infty\)-robust one-step update is
\begin{equation}
\label{eq:clean_centered_exact_delta_step_module}
\delta^{+}
=
\operatorname{clip}\!
\Big(
\delta+\eta_{\delta,t}\,\operatorname{sign}(g_{\delta,t}(v)),
-\epsilon,
+\epsilon
\Big),
\qquad
u^{+}=u-\eta_{u,t}\,g_{u,t}(\delta^{+},u),
\end{equation}
where $x\in\mathbb{R}^{m}$ is the clean input and $v=(\delta,u)$ is the coupled state.
We set $\operatorname{sign}(0):=0$. The dead-zone conditions introduced below keep the relevant normalized gradient coordinates away from zero, so the approximation budgets are determined entirely by those nonzero coordinates.
In the polynomial update model, $g_{\delta,t}$ and $g_{u,t}$ are replaced by polynomial approximations, $\operatorname{sign}$ is replaced by an odd polynomial $P_{\mathrm s}$, and the saturation map $\operatorname{clip}(\cdot,-1,1)$ is replaced by an odd polynomial $P_{\mathrm c}$. The main theorem in this subsection is therefore a theorem about the full-step polynomial update model defined by \eqref{eq:delta_folded_map}--\eqref{eq:full_step_vector_fields_clean}. The transfer to the exact nonsmooth dynamics is stated later in Corollary~\ref{cor:exact_dynamics_bridge_clean}.

Each tensor-power space \(\mathbb{R}^{d^{j}}\) is equipped with the Euclidean norm.
The lifted space
\[
\mathcal H_N:=\bigoplus_{j=1}^{N}\mathbb{R}^{d^j}
\]
uses the direct-sum Euclidean norm
\(\|y\|_2^2=\sum_{j=1}^{N}\|y_j\|_2^2\).
The horizon space \(\mathbb{R}^{(T+1)\Delta_N}\), where
\[
\Delta_N:=\sum_{j=1}^{N} d^j,
\qquad
N_h:=(T+1)\Delta_N,
\]
likewise uses the Euclidean norm
\(\|Y\|_2^2=\sum_{t=0}^{T}\|y(t)\|_2^2\).

\begin{proposition}[Effective degree of the polynomial update]
\label{prop:folded_effective_degree_clean}
Let \(q\ge 1\), and let \(P_{\mathrm s}\) and \(P_{\mathrm c}\) be odd polynomials of degrees
\(K_s\) and \(K_c\), respectively.
Then the polynomial perturbation map \(F_t\) in \eqref{eq:delta_folded_map} has degree at most
\begin{equation}
\label{eq:D_delta_clean_module}
D_{\delta}\le qK_sK_c,
\end{equation}
and the full one-step map in \eqref{eq:full_step_vector_fields_clean} admits an expansion
\begin{equation}
\label{eq:full_step_field_expansions_clean}
\Psi_t(v)=\sum_{\ell=0}^{D}Q_{\ell}^{(t)}v^{\otimes \ell}
\end{equation}
with
\begin{equation}
\label{eq:D_full_clean_module}
D\le qD_{\delta}\le q^2K_sK_c.
\end{equation}
When the outer update is evaluated at the current perturbation \(\delta\), the same argument gives the smaller bound \(D\le \max\{D_{\delta},q\}\).
\end{proposition}

\begin{proof}
Each coordinate of \(\mathcal G_{\delta,t}(v)\) and \(\mathcal G_{u,t}(v)\) has degree at most \(q\).
Hence \(P_{\mathrm s}(\mathcal G_{\delta,t}(v)/\alpha_t)\) has degree at most \(qK_s\).
Therefore the argument of \(P_{\mathrm c}\) in \eqref{eq:delta_folded_map} has degree at most \(qK_s\), and composing with the degree-\(K_c\) polynomial \(P_{\mathrm c}\) gives
\[
D_\delta\le qK_sK_c.
\]
This proves \eqref{eq:D_delta_clean_module}.

Now define \(H_t(v):=(F_t(v),u)\). Its first block has degree at most \(D_\delta\), while its second block is linear, so \(H_t\) has degree at most \(D_\delta\).
Since \(\mathcal G_{u,t}\) has degree at most \(q\) in its state argument, the composition
\(\mathcal G_{u,t}(F_t(v),u)=\mathcal G_{u,t}(H_t(v))\)
has degree at most \(qD_\delta\).
Hence the second block of \(\Psi_t\) has degree at most \(qD_\delta\), whereas the first block already has degree at most \(D_\delta\). Therefore
\[
D\le qD_\delta\le q^2K_sK_c,
\]
which proves \eqref{eq:D_full_clean_module} and \eqref{eq:full_step_field_expansions_clean}.
When the outer update is evaluated at the current perturbation \(\delta\), the composition with \(F_t\) disappears in the second block, and the same argument gives \(D\le \max\{D_\delta,q\}\).
\end{proof}

\subsubsection{Sign and clip approximation}
\label{subsubsec:jackson_chebyshev_folded_clean}

The main theorem in \Cref{thm:one_shot_folded_clean} treats the sign degree $K_s$, the clip degree $K_c$, and the effective full-step degree $D\le q^2K_sK_c$ as design parameters. The next proposition combines a standard odd polynomial approximation of the sign function, the shifted-sign clip construction, the resulting one-step perturbation estimate, and the corresponding degree choice.

\begin{proposition}[Polynomial sign/clip design and one-step perturbation approximation]
\label{prop:sign_clip_design_clean}
There exist absolute constants $C_{\mathrm s},c_{\mathrm s}>0$ such that the following hold.
\begin{enumerate}
\item \textbf{Scaled gapped-sign approximation.}
For every $L>0$, $\tau\in(0,L)$, and $\delta\in(0,1/2)$, there exists an odd polynomial
$P_{\mathrm s}^{(L,\tau,\delta)}$ of degree at most
\begin{equation}
\label{eq:scaled_sign_degree_clean}
K_s
\;\le\;
C_{\mathrm s}\,\frac{L}{\tau}
\log\!\Big(\frac{c_{\mathrm s}}{\delta}\Big)
\end{equation}
such that
\begin{equation}
\label{eq:scaled_sign_spec_clean}
|P_{\mathrm s}^{(L,\tau,\delta)}(x)|\le 1\quad\text{for all }x\in[-L,L],
\qquad
\sup_{|x|\in[\tau,L]}
\big|P_{\mathrm s}^{(L,\tau,\delta)}(x)-\operatorname{sign}(x)\big|
\le \delta.
\end{equation}

\item \textbf{Shifted-sign clip design.}
Fix $L_c>1$, $\tau_c\in(0,L_c-1]$, and $\delta_c\in(0,1)$, and set
\begin{equation}
\label{eq:Rc_clean}
R_c:=L_c+1.
\end{equation}
Let $S$ be an odd polynomial on $[-R_c,R_c]$ such that
\begin{equation}
\label{eq:shifted_sign_assumption_clean}
|S(z)|\le 1\quad\text{for all }|z|\le R_c,
\qquad
\sup_{|z|\in[\tau_c,R_c]}
\big|S(z)-\operatorname{sign}(z)\big|
\le \frac{\delta_c}{L_c}.
\end{equation}
Define
\begin{equation}
\label{eq:clip_poly_from_sign_clean}
P_{\mathrm c}(x)
:=
\frac{1}{2}\Big((x+1)S(x+1)-(x-1)S(x-1)\Big).
\end{equation}
Then $P_{\mathrm c}$ is an odd polynomial of degree at most $\deg S+1$, and
\begin{equation}
\label{eq:clip_poly_inner_clean}
|P_{\mathrm c}(x)-x|\le \delta_c
\qquad\text{for every }|x|\le 1-\tau_c,
\end{equation}
\begin{equation}
\label{eq:clip_poly_outer_clean}
|P_{\mathrm c}(x)-\operatorname{sign}(x)|\le \delta_c
\qquad\text{for every }1+\tau_c\le |x|\le L_c,
\end{equation}
and
\begin{equation}
\label{eq:clip_poly_bounded_clean}
|P_{\mathrm c}(x)|\le 1
\qquad\text{for every }x\in[-1,1].
\end{equation}
If in addition
\begin{equation}
\label{eq:clip_degree_small_error_clean}
\frac{\delta_c}{L_c}\in(0,1/2),
\end{equation}
then one may choose $S$ from item~(1) on $[-R_c,R_c]$, and the resulting clip polynomial satisfies
\begin{equation}
\label{eq:clip_degree_choice_clean}
K_c
\;\le\;
1+C_{\mathrm s}\,\frac{L_c+1}{\tau_c}
\log\!\Big(\frac{c_{\mathrm s}L_c}{\delta_c}\Big).
\end{equation}
In particular, if $L_c=\mathcal O(1)$ on the local domain, then
\begin{equation}
\label{eq:clip_degree_choice_O_clean}
K_c
=
\mathcal O\!\Big(\tau_c^{-1}\log(1/\delta_c)\Big).
\end{equation}

\item \textbf{One-step perturbation approximation.}
Fix a state $v=(\delta,u)$. Assume that each normalized gradient coordinate satisfies
\begin{equation}
\label{eq:dead_zone_condition_clean}
\left|\frac{\mathcal G_{\delta,i}(v)}{\alpha}\right|\in[\tau_s,1],
\qquad i=1,\dots,m,
\end{equation}
and let $P_{\mathrm s}$ be an odd polynomial satisfying
\begin{equation}
\label{eq:sign_poly_condition_clean}
\sup_{|x|\in[\tau_s,1]}|P_{\mathrm s}(x)-\operatorname{sign}(x)|\le \delta_s,
\qquad
|P_{\mathrm s}(x)|\le 1\quad\text{for all }x\in[-1,1].
\end{equation}
Assume also that a clip polynomial $P_{\mathrm c}$ satisfies the conclusions of item~(2) on a domain $[-L_c,L_c]$, and that the polynomial-step argument avoids the transition band:
\begin{equation}
\label{eq:clip_safe_region_clean}
\left|
\frac{\delta_i+\eta_{\delta}P_{\mathrm s}(\mathcal G_{\delta,i}(v)/\alpha)}{\epsilon}
\right|
\in [0,1-\tau_c]\cup[1+\tau_c,L_c],
\qquad i=1,\dots,m.
\end{equation}
Define the exact sign/clip step and the polynomial sign/clip step coordinate-wise by
\begin{align}
\delta_i^{\mathrm{ex}}
&:=
\epsilon\,\operatorname{sat}\!\left(
\frac{\delta_i+\eta_{\delta}\operatorname{sign}(\mathcal G_{\delta,i}(v)/\alpha)}{\epsilon}
\right),
\label{eq:delta_exact_step_clean}
\\
\delta_i^{\mathrm{poly}}
&:=
\epsilon\,P_{\mathrm c}\!\left(
\frac{\delta_i+\eta_{\delta}P_{\mathrm s}(\mathcal G_{\delta,i}(v)/\alpha)}{\epsilon}
\right).
\label{eq:delta_poly_step_clean}
\end{align}
Then
\begin{equation}
\label{eq:onestep_inf_bound_clean}
\|\delta^{\mathrm{ex}}-\delta^{\mathrm{poly}}\|_\infty
\le
\eta_\delta\,\delta_s + \epsilon\,\delta_c,
\end{equation}
and therefore
\begin{equation}
\label{eq:onestep_l2_bound_clean}
\|\delta^{\mathrm{ex}}-\delta^{\mathrm{poly}}\|_2
\le
\sqrt m\,\bigl(\eta_\delta\,\delta_s + \epsilon\,\delta_c\bigr).
\end{equation}

\item \textbf{Choosing $(\delta_s,\delta_c)$ and $(K_s,K_c)$ from a one-step budget.}
Fix a target one-step nonlinear approximation budget
$\varepsilon_{\mathrm{nl,step}}>0$ in $\ell_2$ norm. Under the standing assumptions of item~(3), assume further that
\begin{equation}
\label{eq:budget_small_error_regime_clean}
0<\varepsilon_{\mathrm{nl,step}}<\min\!\bigl\{\eta_{\delta}\sqrt m,\,\epsilon L_c\sqrt m\bigr\}.
\end{equation}
It suffices to choose
\begin{equation}
\label{eq:delta_s_delta_c_choice_clean}
\delta_s:=\frac{\varepsilon_{\mathrm{nl,step}}}{2\eta_{\delta}\sqrt m},
\qquad
\delta_c:=\frac{\varepsilon_{\mathrm{nl,step}}}{2\epsilon\sqrt m}.
\end{equation}
Then
\begin{equation}
\label{eq:one_step_budget_result_clean}
\|\delta^{\mathrm{ex}}-\delta^{\mathrm{poly}}\|_2
\le
\varepsilon_{\mathrm{nl,step}}.
\end{equation}
Moreover one may choose the sign degree and clip degree so that
\begin{align}
K_s
&\le
C_{\mathrm s}\,\tau_s^{-1}
\log\!\Big(\frac{2c_{\mathrm s}\eta_{\delta}\sqrt m}{\varepsilon_{\mathrm{nl,step}}}\Big),
\label{eq:Ks_budget_clean}
\\
K_c
&\le
1+C_{\mathrm s}\,\frac{L_c+1}{\tau_c}
\log\!\Big(\frac{4c_{\mathrm s}L_c\epsilon\sqrt m}{\varepsilon_{\mathrm{nl,step}}}\Big).
\label{eq:Kc_budget_clean}
\end{align}
Consequently, the perturbation-step degree satisfies
\begin{equation}
\label{eq:D_budget_clean}
D_\delta\le qK_sK_c,
\end{equation}
and the full-step degree obeys \(D\le qD_\delta\le q^2K_sK_c\).
\end{enumerate}
\end{proposition}

\begin{proof}
\textbf{Item~(1).}
We use a standard odd bounded polynomial approximation result for the gapped sign function and do not reprove it here. Apply that result on $[-1,1]$ with gap $\widetilde\tau:=\tau/L\in(0,1)$ and accuracy $\delta$, obtaining an odd polynomial $Q$ of degree at most
\[
C_{\mathrm s}\,\widetilde\tau^{-1}\log\!\Big(\frac{c_{\mathrm s}}{\delta}\Big)
=
C_{\mathrm s}\,\frac{L}{\tau}\log\!\Big(\frac{c_{\mathrm s}}{\delta}\Big).
\]
Define
\begin{equation}
\label{eq:scaled_sign_def_clean}
P_{\mathrm s}^{(L,\tau,\delta)}(x):=Q(x/L).
\end{equation}
Because scaling by the positive number $L$ preserves oddness and the sign, this polynomial satisfies \eqref{eq:scaled_sign_degree_clean}--\eqref{eq:scaled_sign_spec_clean}.

\medskip\noindent\textbf{Item~(2).}
We first record the shifted-sign identity
\begin{equation}
\label{eq:sat_shifted_sign_clean}
\operatorname{sat}(x)
=
\frac{1}{2}\Big((x+1)\operatorname{sign}(x+1)-(x-1)\operatorname{sign}(x-1)\Big),
\end{equation}
which is checked by separating the regions $x>1$, $x=1$, $-1<x<1$, $x=-1$, and $x<-1$.
Because $S$ is odd, the polynomial $P_{\mathrm c}$ defined by \eqref{eq:clip_poly_from_sign_clean} is odd and has degree at most $\deg S+1$.
If $|x|\le 1-\tau_c$, then $x+1\in[\tau_c,2-\tau_c]$ and $x-1\in[-2+\tau_c,-\tau_c]$, so \eqref{eq:shifted_sign_assumption_clean} and \eqref{eq:sat_shifted_sign_clean} give
\[
P_{\mathrm c}(x)-x
=
\frac12\Big((x+1)(S(x+1)-1)-(x-1)(S(x-1)+1)\Big),
\]
hence \eqref{eq:clip_poly_inner_clean}. If $1+\tau_c\le x\le L_c$, then
\[
P_{\mathrm c}(x)-1
=
\frac12\Big((x+1)(S(x+1)-1)-(x-1)(S(x-1)-1)\Big),
\]
so \eqref{eq:clip_poly_outer_clean} follows for positive $x$, and oddness gives the negative side. Finally, if $|x|\le 1$, then $|x\pm1|\le 2\le R_c$, hence \eqref{eq:shifted_sign_assumption_clean} implies
\[
|P_{\mathrm c}(x)|
\le
\frac12\Big((x+1)|S(x+1)|+(1-x)|S(x-1)|\Big)\le 1,
\]
which is \eqref{eq:clip_poly_bounded_clean}.
Under \eqref{eq:clip_degree_small_error_clean}, apply item~(1) on $[-R_c,R_c]$ with gap $\tau_c$ and accuracy $\delta_c/L_c$.
This yields an odd polynomial $S$ with
\[
\deg S\le C_{\mathrm s}\,\frac{L_c+1}{\tau_c}
\log\!\Big(\frac{c_{\mathrm s}L_c}{\delta_c}\Big),
\]
and substituting that $S$ into \eqref{eq:clip_poly_from_sign_clean} proves \eqref{eq:clip_degree_choice_clean}. The order estimate \eqref{eq:clip_degree_choice_O_clean} is the immediate simplification when $L_c=\mathcal O(1)$.

\medskip\noindent\textbf{Item~(3).}
For each coordinate $i$, define
\[
z_i:=
\frac{\delta_i+\eta_{\delta}\operatorname{sign}(\mathcal G_{\delta,i}(v)/\alpha)}{\epsilon},
\qquad
\widetilde z_i:=
\frac{\delta_i+\eta_{\delta}P_{\mathrm s}(\mathcal G_{\delta,i}(v)/\alpha)}{\epsilon}.
\]
Then \eqref{eq:delta_exact_step_clean}--\eqref{eq:delta_poly_step_clean} become
\[
\delta_i^{\mathrm{ex}}=\epsilon\,\operatorname{sat}(z_i),
\qquad
\delta_i^{\mathrm{poly}}=\epsilon\,P_{\mathrm c}(\widetilde z_i).
\]
Using the triangle inequality,
\[
|\delta_i^{\mathrm{ex}}-\delta_i^{\mathrm{poly}}|
\le
\epsilon\,|\operatorname{sat}(z_i)-\operatorname{sat}(\widetilde z_i)|
+
\epsilon\,|\operatorname{sat}(\widetilde z_i)-P_{\mathrm c}(\widetilde z_i)|.
\]
The scalar saturation map is $1$-Lipschitz, so \eqref{eq:dead_zone_condition_clean} and \eqref{eq:sign_poly_condition_clean} imply
\begin{equation}
\label{eq:sign_part_bound_clean}
\epsilon\,|\operatorname{sat}(z_i)-\operatorname{sat}(\widetilde z_i)|
\le
\eta_\delta\,\delta_s.
\end{equation}
By \eqref{eq:clip_safe_region_clean}, the argument $\widetilde z_i$ lies either in the interior region of \eqref{eq:clip_poly_inner_clean} or in the outer region of \eqref{eq:clip_poly_outer_clean}, and in either case
\begin{equation}
\label{eq:clip_part_bound_clean}
\epsilon\,|\operatorname{sat}(\widetilde z_i)-P_{\mathrm c}(\widetilde z_i)|
\le
\epsilon\,\delta_c.
\end{equation}
Adding \eqref{eq:sign_part_bound_clean} and \eqref{eq:clip_part_bound_clean} proves \eqref{eq:onestep_inf_bound_clean}, and \eqref{eq:onestep_l2_bound_clean} follows from
\[
\|\delta^{\mathrm{ex}}-\delta^{\mathrm{poly}}\|_2
\le
\sqrt m\,\|\delta^{\mathrm{ex}}-\delta^{\mathrm{poly}}\|_\infty.
\]

\medskip\noindent\textbf{Item~(4).}
Substituting \eqref{eq:delta_s_delta_c_choice_clean} into \eqref{eq:onestep_l2_bound_clean} gives
\[
\sqrt m\big(\eta_{\delta}\delta_s+\epsilon\delta_c\big)
=
\varepsilon_{\mathrm{nl,step}},
\]
which proves \eqref{eq:one_step_budget_result_clean}. The small-error regime \eqref{eq:budget_small_error_regime_clean} ensures $\delta_s\in(0,1/2)$ and $\delta_c/L_c\in(0,1/2)$. Therefore item~(1) with $(L,\tau,\delta)=(1,\tau_s,\delta_s)$ yields \eqref{eq:Ks_budget_clean}, item~(2) yields \eqref{eq:Kc_budget_clean}, and \eqref{eq:D_budget_clean} is the degree estimate already proved in \Cref{prop:folded_effective_degree_clean}.
\end{proof}

Item~(1) uses only a standard odd polynomial approximation of the sign function. Any odd polynomial family satisfying the same degree and approximation bounds is sufficient for the rest of the argument. For classical background on Jackson/Chebyshev approximation on compact intervals, see DeVore~\cite{devore1968jackson} and Trefethen~\cite{trefethen2019atap}; for sign approximation on symmetric gapped intervals, see Eremenko and Yuditskii~\cite{eremenko2007signApprox}.

Item~(4) is stated for a single step. In the later transfer result and the full-window bounds the step sizes and normalizations may depend on \(t\), while the sign and clip polynomials are fixed once for the whole training window. To enforce one common pair \((K_s,K_c)\), it suffices to apply Item~(4) of \Cref{prop:sign_clip_design_clean} with the same local-domain parameters \(\tau_s,\tau_c,L_c\) and with
\[
\eta_\delta^{\max}:=\sup_{0\le t\le T-1}\eta_{\delta,t}.
\]
Then the one-step estimate \eqref{eq:onestep_l2_bound_clean} gives, for every admissible time \(t\),
\[
\|\delta_t^{\mathrm{ex}}-\delta_t^{\mathrm{poly}}\|_2
\le
\sqrt m\bigl(\eta_{\delta,t}\delta_s+\epsilon\delta_c\bigr)
\le
\sqrt m\bigl(\eta_\delta^{\max}\delta_s+\epsilon\delta_c\bigr)
\le
\varepsilon_{\mathrm{nl,step}},
\]
whenever the dead-zone and clip-safe conditions hold uniformly on the common local domain. The time dependence of \(\alpha_t\) causes no extra degree growth because the sign polynomial is always applied to the normalized quantity \(\mathcal G_{\delta,t}(v)/\alpha_t\in[-1,1]\).
Theorem~\ref{thm:one_shot_folded_clean} itself remains a theorem for the polynomial update model. Proposition~\ref{prop:sign_clip_design_clean} and Proposition~\ref{prop:base_to_lifted_bridge_clean} enter only when transferring that theorem to the exact adversarial-training dynamics in Corollary~\ref{cor:exact_dynamics_bridge_clean}.

\begin{proposition}[From one-step approximation error to the lifted error term]
\label{prop:base_to_lifted_bridge_clean}
Let $\Psi_t^{\mathrm{phys}}$ be the exact robust one-step map from
\eqref{eq:clean_centered_exact_delta_step_module}, and let $\Psi_t^{\mathrm{fold}}$ be the full-step polynomial update model from
\eqref{eq:full_step_update_clean}. Then the following hold.
\begin{enumerate}
\item \textbf{From perturbation-step error to a one-step state error.}
Fix a local domain on which the following hold uniformly for every admissible $v=(\delta,u)$ and every time $t$.
\begin{enumerate}
\item[(i)] There exists $\varepsilon_{\delta,\mathrm{grad}}\ge 0$ such that
\[
\|g_{\delta,t}(v)-\mathcal G_{\delta,t}(v)\|_\infty
\le
\varepsilon_{\delta,\mathrm{grad}},
\]
the normalized polynomial approximation stays away from the origin coordinate-wise,
\[
\left|\frac{\mathcal G_{\delta,i,t}(v)}{\alpha_t}\right|\in[\tau_s,1],
\qquad i=1,\dots,m,
\]
and
\begin{equation}
\label{eq:sign_consistency_margin_clean}
\varepsilon_{\delta,\mathrm{grad}}\le \frac{\alpha_t\tau_s}{2}.
\end{equation}
\item[(ii)] The clip-safe argument satisfies \eqref{eq:clip_safe_region_clean}.
\item[(iii)] There exists $\varepsilon_{u,\mathrm{grad}}\ge 0$ such that
\[
\|g_{u,t}(z)-\mathcal G_{u,t}(z)\|_2
\le
\varepsilon_{u,\mathrm{grad}}
\]
for every admissible state vector $z=(\delta,u)$ reached by either the exact perturbation step or the polynomial perturbation step.
\item[(iv)] There exists $L_{u,\delta}^{(t)}\ge 0$ such that
\[
\|g_{u,t}(\delta_a,u)-g_{u,t}(\delta_b,u)\|_2
\le
L_{u,\delta}^{(t)}\,\|\delta_a-\delta_b\|_2
\]
for every admissible pair $(\delta_a,u)$ and $(\delta_b,u)$.
\end{enumerate}
Then the one-step state error obeys
\begin{equation}
\label{eq:base_step_budget_sufficient_clean}
\|\Psi_t^{\mathrm{phys}}(v)-\Psi_t^{\mathrm{fold}}(v)\|_2
\le
\bigl(1+\eta_{u,t}L_{u,\delta}^{(t)}\bigr)
\varepsilon_{\delta,\mathrm{step}}
+
\eta_{u,t}\,\varepsilon_{u,\mathrm{grad}},
\end{equation}
where
\begin{equation}
\label{eq:delta_step_budget_symbol_clean}
\varepsilon_{\delta,\mathrm{step}}
:=
\|\delta_{\mathrm{phys}}^{+}-\delta_{\mathrm{fold}}^{+}\|_2.
\end{equation}
In particular, if Item~(4) of \Cref{prop:sign_clip_design_clean} is invoked so that
\(
\varepsilon_{\delta,\mathrm{step}}\le \varepsilon_{\mathrm{nl,step}}
\),
then
\begin{equation}
\label{eq:base_step_budget_nonlinearity_clean}
\|\Psi_t^{\mathrm{phys}}(v)-\Psi_t^{\mathrm{fold}}(v)\|_2
\le
\bigl(1+\eta_{u,t}L_{u,\delta}^{(t)}\bigr)
\varepsilon_{\mathrm{nl,step}}
+
\eta_{u,t}\,\varepsilon_{u,\mathrm{grad}}.
\end{equation}
Therefore the condition \eqref{eq:base_step_bridge_assumption_clean} holds with any uniform choice of
\begin{equation}
\label{eq:base_step_budget_uniform_clean}
\varepsilon_{\mathrm{base,step}}
\ge
\sup_{0\le t\le T-1}
\left[
\bigl(1+\eta_{u,t}L_{u,\delta}^{(t)}\bigr)
\varepsilon_{\mathrm{nl,step}}
+
\eta_{u,t}\,\varepsilon_{u,\mathrm{grad}}
\right].
\end{equation}

\item \textbf{From one-step state error to the lifted error term.}
Define the retained Carleman lift
\[
\mathcal L_N(v):=\bigl(v,v^{\otimes 2},\dots,v^{\otimes N}\bigr)\in\mathbb R^{\Delta_N},
\qquad
\Delta_N:=\sum_{j=1}^N d^j.
\]
Fix $\bar v>0$ and assume $\|a\|_2,\|b\|_2\le \bar v$. Then
\begin{equation}
\label{eq:lift_local_lipschitz_clean}
\|\mathcal L_N(a)-\mathcal L_N(b)\|_2
\le
L_{N,\bar v}^{\mathrm{lift}}\,\|a-b\|_2,
\end{equation}
where
\begin{equation}
\label{eq:lift_local_lipschitz_constant_clean}
L_{N,\bar v}^{\mathrm{lift}}
:=
\left(\sum_{j=1}^N j^2\bar v^{2j-2}\right)^{1/2}.
\end{equation}
Now let $\mathcal V\subseteq\mathbb R^d$ satisfy
\[
\mathcal V\subseteq \{v\in\mathbb R^d:\|v\|_2\le \bar v<1\},
\]
and assume that for every $v\in\mathcal V$ and every $t=0,1,\dots,T-1$ one has
\begin{equation}
\label{eq:base_step_bridge_assumption_clean}
\|\Psi_t^{\mathrm{phys}}(v)-\Psi_t^{\mathrm{fold}}(v)\|_2
\le
\varepsilon_{\mathrm{base,step}},
\end{equation}
and assume
\begin{equation}
\label{eq:bridge_image_invariance_clean}
\Psi_t^{\mathrm{phys}}(v)\in\mathcal V,
\qquad
\Psi_t^{\mathrm{fold}}(v)\in\mathcal V.
\end{equation}
Then
\begin{equation}
\label{eq:lifted_model_step_bound_clean}
\|\mathcal L_N(\Psi_t^{\mathrm{phys}}(v)) - \mathcal L_N(\Psi_t^{\mathrm{fold}}(v))\|_2
\le
L_{N,\bar v}^{\mathrm{lift}}\,\varepsilon_{\mathrm{base,step}}.
\end{equation}
In particular, if the exact truncated lifted dynamics is written in the form
\begin{equation}
\label{eq:model_forcing_decomposition_clean}
y_{\mathrm{phys}}(t+1)
=
B(t)\,y_{\mathrm{phys}}(t)+c(t)+e_{\mathrm{tr}}(t)+e_{\mathrm{model}}(t),
\end{equation}
where $e_{\mathrm{model}}(t)$ is exactly the discrepancy between the exact next-step truncated lift and the next-step truncated lift produced by the polynomial update model, evaluated on the same current state, then
\begin{equation}
\label{eq:model_forcing_step_bound_clean}
\|e_{\mathrm{model}}(t)\|_2
\le
L_{N,\bar v}^{\mathrm{lift}}\,\varepsilon_{\mathrm{base,step}}
\qquad
\text{for all }t=0,1,\dots,T-1.
\end{equation}
\end{enumerate}
\end{proposition}

\begin{proof}
\textbf{Item~(1).}
We split the proof into perturbation and parameter blocks.

\medskip\noindent\textbf{Step 1. The exact and polynomial inner signs agree.}
Fix a coordinate $i$. By the dead-zone hypothesis,
\(
|\mathcal G_{\delta,i,t}(v)|\ge \alpha_t\tau_s
\).
Together with \eqref{eq:sign_consistency_margin_clean}, this gives two cases. If \(\mathcal G_{\delta,i,t}(v)\ge \alpha_t\tau_s\), then
\[
g_{\delta,i,t}(v)
\ge
\mathcal G_{\delta,i,t}(v)-|g_{\delta,i,t}(v)-\mathcal G_{\delta,i,t}(v)|
\ge
\alpha_t\tau_s-\varepsilon_{\delta,\mathrm{grad}}
\ge
\frac{\alpha_t\tau_s}{2}>0.
\]
If \(\mathcal G_{\delta,i,t}(v)\le -\alpha_t\tau_s\), then
\[
g_{\delta,i,t}(v)
\le
\mathcal G_{\delta,i,t}(v)+|g_{\delta,i,t}(v)-\mathcal G_{\delta,i,t}(v)|
\le
-\alpha_t\tau_s+\varepsilon_{\delta,\mathrm{grad}}
\le
-\frac{\alpha_t\tau_s}{2}<0.
\]
Hence $g_{\delta,i,t}(v)$ and $\mathcal G_{\delta,i,t}(v)$ have the same sign for every coordinate. Therefore the exact perturbation step from \eqref{eq:clean_centered_exact_delta_step_module} is exactly the sign/clip step treated in Item~(3) of \Cref{prop:sign_clip_design_clean}, with the same polynomial approximation built from $\mathcal G_{\delta,t}$.

\medskip\noindent\textbf{Step 2. Perturbation-step bound.}
By definition,
\[
\varepsilon_{\delta,\mathrm{step}}
=
\|\delta_{\mathrm{phys}}^{+}-\delta_{\mathrm{fold}}^{+}\|_2.
\]
This is exactly the quantity that appears in \eqref{eq:base_step_budget_sufficient_clean}. If, in addition, the sign and clip polynomials are chosen through Item~(4) of \Cref{prop:sign_clip_design_clean} on the same local domain, then Step~1 identifies the perturbation update with the sign/clip setting of Item~(3) of \Cref{prop:sign_clip_design_clean}, and hence
\(
\varepsilon_{\delta,\mathrm{step}}\le \varepsilon_{\mathrm{nl,step}}
\).

\medskip\noindent\textbf{Step 3. Bound the parameter-step discrepancy.}
Write
\[
\Delta_u
:=
 u_{\mathrm{phys}}^{+}-u_{\mathrm{fold}}^{+}
=
-\eta_{u,t}\Bigl(g_{u,t}(\delta_{\mathrm{phys}}^{+},u)-\mathcal G_{u,t}(\delta_{\mathrm{fold}}^{+},u)\Bigr).
\]
Add and subtract $g_{u,t}(\delta_{\mathrm{fold}}^{+},u)$ to obtain
\[
\|\Delta_u\|_2
\le
\eta_{u,t}\,
\|g_{u,t}(\delta_{\mathrm{phys}}^{+},u)-g_{u,t}(\delta_{\mathrm{fold}}^{+},u)\|_2
+
\eta_{u,t}\,
\|g_{u,t}(\delta_{\mathrm{fold}}^{+},u)-\mathcal G_{u,t}(\delta_{\mathrm{fold}}^{+},u)\|_2.
\]
Using hypothesis (iv) for the first term and hypothesis (iii) for the second gives
\begin{equation}
\label{eq:u_step_bridge_bound_clean}
\|\Delta_u\|_2
\le
\eta_{u,t}L_{u,\delta}^{(t)}\,\varepsilon_{\delta,\mathrm{step}}
+
\eta_{u,t}\,\varepsilon_{u,\mathrm{grad}}.
\end{equation}

\medskip\noindent\textbf{Step 4. Assemble the full one-step bound.}
The full-state discrepancy is
\[
\Psi_t^{\mathrm{phys}}(v)-\Psi_t^{\mathrm{fold}}(v)
=
\bigl(\delta_{\mathrm{phys}}^{+}-\delta_{\mathrm{fold}}^{+},\;u_{\mathrm{phys}}^{+}-u_{\mathrm{fold}}^{+}\bigr).
\]
Hence
\[
\|\Psi_t^{\mathrm{phys}}(v)-\Psi_t^{\mathrm{fold}}(v)\|_2
\le
\varepsilon_{\delta,\mathrm{step}}+\|\Delta_u\|_2.
\]
Substituting \eqref{eq:u_step_bridge_bound_clean} proves \eqref{eq:base_step_budget_sufficient_clean}. The remaining displayed inequalities are immediate consequences.

\medskip\noindent\textbf{Item~(2).}
For each $j\in\{1,\dots,N\}$, the tensor-power difference admits the telescoping expansion
\[
a^{\otimes j}-b^{\otimes j}
=
\sum_{r=1}^j a^{\otimes(r-1)}\otimes(a-b)\otimes b^{\otimes(j-r)}.
\]
Taking Euclidean norms and using submultiplicativity of tensor-product norms gives
\[
\|a^{\otimes j}-b^{\otimes j}\|_2
\le
\sum_{r=1}^j \|a\|_2^{r-1}\,\|a-b\|_2\,\|b\|_2^{j-r}
\le
j\bar v^{j-1}\|a-b\|_2.
\]
Squaring and summing over $j=1,\dots,N$ yields \eqref{eq:lift_local_lipschitz_clean}. Now apply \eqref{eq:lift_local_lipschitz_clean} with
\[
a=\Psi_t^{\mathrm{phys}}(v),
\qquad
b=\Psi_t^{\mathrm{fold}}(v).
\]
By \eqref{eq:base_step_bridge_assumption_clean},
\[
\|a-b\|_2\le \varepsilon_{\mathrm{base,step}},
\]
and by \eqref{eq:bridge_image_invariance_clean}, both images belong to $\mathcal V\subseteq\{z:\|z\|_2\le \bar v\}$. Therefore
\[
\|\mathcal L_N(\Psi_t^{\mathrm{phys}}(v)) - \mathcal L_N(\Psi_t^{\mathrm{fold}}(v))\|_2
\le
L_{N,\bar v}^{\mathrm{lift}}\,\varepsilon_{\mathrm{base,step}},
\]
which is \eqref{eq:lifted_model_step_bound_clean}. The bound \eqref{eq:model_forcing_step_bound_clean} is then just the definition of the modeling error term $e_{\mathrm{model}}(t)$ in \eqref{eq:model_forcing_decomposition_clean}.
\end{proof}

\subsubsection{Discrete-time Carleman lift for the full-step recurrence}
\label{subsubsec:clean_carleman_switched}

For each time \(t\), consider the full-step polynomial map
\[
\Psi_t(v)=\sum_{\ell=0}^{D}Q_{\ell}^{(t)}v^{\otimes \ell}
\qquad
\text{on }\mathbb R^d.
\]
Let
\[
y_0(t):=1,
\qquad
y_j(t):=v(t)^{\otimes j}\ \ (j\ge 1),
\qquad
y^{(N)}(t):=\big(y_1(t),y_2(t),\dots,y_N(t)\big)\in\mathcal H_N.
\]

\begin{proposition}[Discrete-time Carleman lift, time-unrolled system, and conditioning]
\label{prop:carleman_horizon_cond_clean}
Fix \(t\) and consider the polynomial map
\[
\Psi_t(v)=\sum_{\ell=0}^{D}Q_{\ell}^{(t)}v^{\otimes \ell}
\qquad
\text{on }\mathbb R^d.
\]
Then the following hold.
\begin{enumerate}
\item \textbf{Exact discrete-time Carleman blocks.}
For each \(j\ge 1\) and each \(s\in\{0,1,\dots,jD\}\), define the discrete-time Carleman block
\begin{equation}
\label{eq:discrete_carleman_blocks_clean}
K_{j,s}(t)
:=
\sum_{\substack{\alpha=(\alpha_1,\dots,\alpha_j)\in\{0,1,\dots,D\}^j\\|\alpha|=s}}
Q_{\alpha_1}^{(t)}\otimes Q_{\alpha_2}^{(t)}\otimes\cdots\otimes Q_{\alpha_j}^{(t)},
\qquad
|\alpha|:=\alpha_1+\cdots+\alpha_j.
\end{equation}
Then the lifted monomials satisfy the exact recurrence
\begin{equation}
\label{eq:lifted_component_equation_clean}
y_j(t+1)
=
\sum_{s=0}^{jD}K_{j,s}(t)\,y_s(t),
\qquad j\ge 1.
\end{equation}

\item \textbf{Time-unrolled horizon system.}
Fix a Carleman cutoff \(N\ge 1\). For \(j=1,\dots,N\), define
\[
c_j(t):=K_{j,0}(t)\in\mathbb R^{d^j},
\qquad
c(t):=\big(c_1(t),c_2(t),\dots,c_N(t)\big)\in\mathcal H_N,
\]
let \(B(t)\in\mathbb R^{\Delta_N\times\Delta_N}\) be the truncated lifted step matrix whose \((j,s)\)-block is
\[
B_{j,s}(t):=K_{j,s}(t),
\qquad
1\le j,s\le N,
\]
and define the truncated lifted recurrence
\begin{equation}
\label{eq:truncated_lifted_euler_clean}
\hat y(t+1)=B(t)\hat y(t)+c(t),
\qquad
\hat y(0)=y^{(N)}(0).
\end{equation}
If
\[
Y:=\big(\hat y(0),\hat y(1),\dots,\hat y(T)\big)\in\mathbb{R}^{(T+1)\Delta_N},
\qquad
B_{\mathrm{rhs}}:=\big(\hat y(0),c(0),c(1),\dots,c(T-1)\big),
\]
then \(Y\) solves the block-lower-bidiagonal system
\begin{equation}
\label{eq:horizon_system_clean}
M Y = B_{\mathrm{rhs}},
\qquad
M=
\begin{pmatrix}
I\\
-B(0) & I\\
& -B(1) & I\\
&& \ddots & \ddots\\
&&& -B(T-1) & I
\end{pmatrix},
\end{equation}
and conversely every solution of \eqref{eq:horizon_system_clean} satisfies \eqref{eq:truncated_lifted_euler_clean}.

\item \textbf{Conditioning under contractivity.}
Assume that
\begin{equation}
\label{eq:contractive_B_clean}
\sup_{0\le t\le T-1}\|B(t)\|_2\le \rho<1.
\end{equation}
Then the horizon matrix \(M\) in \eqref{eq:horizon_system_clean} is invertible and satisfies
\begin{equation}
\label{eq:M_norm_bound_clean}
\|M\|_2\le 1+\rho,
\qquad
\|M^{-1}\|_2\le \sum_{k=0}^{T}\rho^k\le \frac{1}{1-\rho},
\end{equation}
and therefore
\begin{equation}
\label{eq:kappa_bound_clean}
\kappa_2(M)
\le
(1+\rho)\sum_{k=0}^{T}\rho^k
\le
\min\!\left\{\frac{1+\rho}{1-\rho},\,2(T+1)\right\}.
\end{equation}
\end{enumerate}
\end{proposition}

\begin{proof}
For item~(1), since \(v(t+1)=\Psi_t(v(t))\), one has
\[
y_j(t+1)=\Psi_t(v(t))^{\otimes j}
=\Big(\sum_{\ell=0}^D Q_\ell^{(t)}v(t)^{\otimes \ell}\Big)^{\otimes j}.
\]
Expanding the \(j\)-fold tensor power and grouping terms by total degree \(s=|\alpha|\) yields
\eqref{eq:discrete_carleman_blocks_clean} and \eqref{eq:lifted_component_equation_clean}.

For item~(2), truncating the lifted recurrence to levels \(1,\dots,N\) gives
\(\hat y(t+1)=B(t)\hat y(t)+c(t)\).
Stacking these equations together with the initialization \(I\hat y(0)=\hat y(0)\) yields the block-lower-bidiagonal system \eqref{eq:horizon_system_clean}. The converse is immediate by reading the block rows of \(MY=B_{\mathrm{rhs}}\).

For item~(3), write \(M=I-L\), where \(L\) is the strict block-subdiagonal shift with
\((LY)_{t+1}=B(t)\hat y(t)\).
Then \(L^{T+1}=0\) and, by \eqref{eq:contractive_B_clean}, \(\|L\|_2\le \rho<1\).
Hence
\[
M^{-1}=(I-L)^{-1}=\sum_{k=0}^{T}L^k,
\qquad
\|M^{-1}\|_2\le \sum_{k=0}^{T}\rho^k\le \frac{1}{1-\rho}.
\]
Moreover \(\|M\|_2\le \|I\|_2+\|L\|_2\le 1+\rho\), so
\[
\kappa_2(M)
\le
(1+\rho)\sum_{k=0}^{T}\rho^k
\le
\min\!\left\{\frac{1+\rho}{1-\rho},\,2(T+1)\right\}.
\]
This proves the stated conditioning bounds.
\end{proof}

\begin{proposition}[Majorant-based sufficient conditions for contractivity]
\label{prop:contractivity_conditions_clean}
For each time $t$ and cutoff $N$, define the nonnegative scalar matrix
\[
\mathsf R^{(N)}(t)\in\mathbb R^{N\times N}
\]
by
\begin{equation}
\label{eq:contractive_majorant_clean}
\mathsf R^{(N)}_{j,s}(t)
:=
\sum_{\substack{\alpha\in\{0,1,\dots,D\}^{j}\\|\alpha|=s}}
\prod_{r=1}^{j}\|Q_{\alpha_r}^{(t)}\|_2,
\qquad 1\le j,s\le N.
\end{equation}
Then the following hold.
\begin{enumerate}
\item \textbf{General majorant criterion.}
The truncated lifted step matrix $B(t)$ satisfies
\begin{equation}
\label{eq:B_majorized_by_R_clean}
\|B(t)\|_2\le \|\mathsf R^{(N)}(t)\|_2.
\end{equation}
In particular, if
\begin{equation}
\label{eq:majorant_contractivity_clean}
\sup_{0\le t\le T-1}\|\mathsf R^{(N)}(t)\|_2\le \rho<1,
\end{equation}
then hypothesis \emph{(H1)} of Theorem~\ref{thm:one_shot_folded_clean} holds.

\item \textbf{Linear-dominant sufficient condition.}
For each $t$, define the diagonal scalar matrix
\begin{equation}
\label{eq:Rlin_clean}
\mathsf R_{\mathrm{lin}}^{(N)}(t)
:=
\operatorname{diag}\bigl(
\|Q_1^{(t)}\|_2,
\|Q_1^{(t)}\|_2^2,
\dots,
\|Q_1^{(t)}\|_2^N
\bigr)
\end{equation}
and the nonlinear remainder majorant
\begin{equation}
\label{eq:Rrem_clean}
\mathsf R_{\mathrm{rem}}^{(N)}(t)
:=
\mathsf R^{(N)}(t)-\mathsf R_{\mathrm{lin}}^{(N)}(t).
\end{equation}
Assume there exist numbers $\gamma\in(0,1)$ and $\sigma\in[0,\gamma)$ such that for every $t$,
\begin{equation}
\label{eq:linear_dominant_contractive_assumption_clean}
\|Q_1^{(t)}\|_2\le 1-\gamma,
\qquad
\|\mathsf R_{\mathrm{rem}}^{(N)}(t)\|_2\le \sigma.
\end{equation}
Then
\begin{equation}
\label{eq:linear_dominant_contractive_conclusion_clean}
\|B(t)\|_2\le 1-\gamma+\sigma<1
\qquad\text{for all }t,
\end{equation}
so \Cref{thm:one_shot_folded_clean} applies with \(\rho:=1-\gamma+\sigma\).
\end{enumerate}
\end{proposition}

\begin{proof}
\textbf{Item~(1).}

Fix $t$. For each block, the definition of $K_{j,s}(t)$ together with the tensor-product norm identity gives
\[
\|B_{j,s}(t)\|_2
=
\|K_{j,s}(t)\|_2
\le
\sum_{\substack{\alpha\in\{0,\dots,D\}^{j}\\|\alpha|=s}}
\|Q_{\alpha_1}^{(t)}\otimes\cdots\otimes Q_{\alpha_j}^{(t)}\|_2
=
\mathsf R^{(N)}_{j,s}(t).
\]
Now let $y=(y_1,\dots,y_N)\in\mathcal H_N$ and define the nonnegative vector
\[
z:=(\|y_1\|_2,\dots,\|y_N\|_2)^\top\in\mathbb R^{N}_{\ge 0}.
\]
For each retained level $j$,
\[
\|(B(t)y)_j\|_2
\le
\sum_{s=1}^{N}\|B_{j,s}(t)\|_2\,\|y_s\|_2
\le
\sum_{s=1}^{N}\mathsf R^{(N)}_{j,s}(t)\,z_s
=
(\mathsf R^{(N)}(t)z)_j.
\]
Squaring and summing over $j$ yields
\[
\|B(t)y\|_2
\le
\|\mathsf R^{(N)}(t)z\|_2
\le
\|\mathsf R^{(N)}(t)\|_2\,\|z\|_2
=
\|\mathsf R^{(N)}(t)\|_2\,\|y\|_2.
\]
Taking the supremum over all nonzero $y$ proves \eqref{eq:B_majorized_by_R_clean}. The final implication is immediate.
\medskip\noindent\textbf{Item~(2).}

By construction,
\(
\mathsf R^{(N)}(t)=\mathsf R_{\mathrm{lin}}^{(N)}(t)+\mathsf R_{\mathrm{rem}}^{(N)}(t)
\).
Since the diagonal entries of \(\mathsf R_{\mathrm{lin}}^{(N)}(t)\) are
\(
\|Q_1^{(t)}\|_2^j
\) with \(1\le j\le N\), one has
\[
\|\mathsf R_{\mathrm{lin}}^{(N)}(t)\|_2
=
\max_{1\le j\le N}\|Q_1^{(t)}\|_2^j
\le
\|Q_1^{(t)}\|_2
\le 1-\gamma.
\]
Therefore,
\[
\|\mathsf R^{(N)}(t)\|_2
\le
\|\mathsf R_{\mathrm{lin}}^{(N)}(t)\|_2
+
\|\mathsf R_{\mathrm{rem}}^{(N)}(t)\|_2
\le 1-\gamma+\sigma.
\]
Applying item~(1) gives
\(
\|B(t)\|_2\le 1-\gamma+\sigma<1
\), as claimed.
\end{proof}
Item~(2) is most natural in the linear-dominant regime. For a purely linear map, $Q_{\ell}^{(t)}=0$ for all $\ell\neq 0,1$, so the condition reduces to the standard linear requirement $\|Q_1^{(t)}\|_2<1$. More generally, this is a local stability condition on a fixed training window: after translating to deviation coordinates and rescaling a small neighborhood, the linear block stays fixed while every coefficient of degree \(\ell\ge 2\) acquires a small-radius factor. If only segmentwise margins are available, see the segmented local-contractivity statement in the Extensions section.

\begin{proposition}[Coefficient sparsity, horizon sparsity, and row-access realization]
\label{prop:sparsity_oracle_clean}
For each time \(t\in\{0,\dots,T-1\}\) and each degree \(\ell\in\{0,\dots,D\}\), let
\(s_{\ell}^{(t)}\) denote a row-sparsity bound for the coefficient tensor \(Q_{\ell}^{(t)}\) when it is viewed as a matrix
\(\mathbb R^{d^\ell}\to \mathbb R^d\).
For \(1\le j,s\le N\), define
\begin{equation}
\label{eq:coeff_sparse_bound_clean}
S_{j,s}^{(t)}
:=
\sum_{\substack{\alpha\in\{0,\dots,D\}^{j}\\ |\alpha|=s}}
\prod_{r=1}^{j}s_{\alpha_r}^{(t)}.
\end{equation}
Then the following hold.
\begin{enumerate}
\item \textbf{Coefficient sparsity implies lifted-step sparsity.}
Every row of the Carleman block \(K_{j,s}(t)\) in \eqref{eq:discrete_carleman_blocks_clean} has at most
\(S_{j,s}^{(t)}\) nonzero entries. Consequently, every row in block row \(j\) of the truncated lifted step matrix
\(B(t)\) has at most
\(\sum_{s=1}^{N} S_{j,s}^{(t)}\)
nonzero entries, and therefore
\begin{equation}
\label{eq:sB_from_coeff_sparsity_clean}
s_B
\le
\max_{0\le t\le T-1}\;
\max_{1\le j\le N}
\sum_{s=1}^{N} S_{j,s}^{(t)}.
\end{equation}
If, in addition, \(s_{\ell}^{(t)}\le s_{\ast}\) for every \(\ell\) and every \(t\), then
\begin{equation}
\label{eq:sB_uniform_coeff_sparsity_clean}
s_B
\le
\max_{1\le j\le N}
\left(
 s_{\ast}^{\,j}\sum_{s=1}^{N} C_{j,s,D}
\right)
\le
\max_{1\le j\le N}
\left(
 s_{\ast}^{\,j}\Bigl(\binom{N+j}{j}-1\Bigr)
\right),
\end{equation}
where
\begin{equation}
\label{eq:CjsD_clean}
C_{j,s,D}:=
\bigl|
\{\alpha\in\{0,\dots,D\}^{j}:|\alpha|=s\}
\bigr|
\le
\binom{s+j-1}{j-1}.
\end{equation}

\item \textbf{Lifted-step sparsity implies horizon sparsity.}
If each truncated lifted step matrix \(B(t)\) is \(s_B\)-row sparse, then the horizon matrix \(M\) in \eqref{eq:horizon_system_clean} has row sparsity
\begin{equation}
\label{eq:sM_clean}
s_{\mathrm{row}}(M)\le s_B+1.
\end{equation}

\item \textbf{Optional row-access realization.}
Assume each truncated lifted step matrix \(B(t)\) is supplied by standard row-sparse access oracles and has row sparsity at most \(s_B\). Let
\[
\overline M:=\frac{1}{1+\rho}M.
\]
Then \(\overline M\) admits standard row-sparse access oracles with row sparsity at most \(s_B+1\). More precisely, one query to a row-sparse access oracle for \(\overline M\) can be implemented using \(O(1)\) arithmetic on the block/time indices together with \(O(1)\) calls to the row-sparse access oracle for the active truncated lifted step matrix \(B(t)\).
\end{enumerate}
\end{proposition}

\begin{proof}
We prove the three items in order.

\medskip\noindent\textbf{Item~(1).}
Fix \(t,j,s\). By \eqref{eq:discrete_carleman_blocks_clean},
\[
K_{j,s}(t)
=
\sum_{\substack{\alpha\in\{0,\dots,D\}^{j}\\ |\alpha|=s}}
Q_{\alpha_1}^{(t)}\otimes\cdots\otimes Q_{\alpha_j}^{(t)}.
\]
Consider one summand \(Q_{\alpha_1}^{(t)}\otimes\cdots\otimes Q_{\alpha_j}^{(t)}\).
A row of this tensor product is indexed by a tuple of output coordinates \((i_1,\dots,i_j)\), and its nonzero locations are Cartesian products of the nonzero locations from the corresponding rows of the factors. Hence that row contains at most
\(
\prod_{r=1}^{j}s_{\alpha_r}^{(t)}
\)
nonzero entries.
Summing over all admissible multi-indices \(\alpha\) and using a union bound on supports gives the row-sparsity estimate \eqref{eq:coeff_sparse_bound_clean} for \(K_{j,s}(t)\).

Now fix a block row \(j\) of the truncated lifted step matrix \(B(t)\). The blocks
\(K_{j,1}(t),\dots,K_{j,N}(t)\) occupy disjoint column ranges, so a row in that block row is obtained by concatenating the corresponding rows of these blocks. Hence
\[
\operatorname{rowsp}\bigl(B(t)\text{ on block row }j\bigr)
\le
\sum_{s=1}^{N} S_{j,s}^{(t)}.
\]
Taking the maximum over retained levels and times yields \eqref{eq:sB_from_coeff_sparsity_clean}.

If \(s_{\ell}^{(t)}\le s_{\ast}\) uniformly, then each product in \eqref{eq:coeff_sparse_bound_clean} is at most \(s_{\ast}^{\,j}\), and the number of admissible multi-indices is \(C_{j,s,D}\). Hence
\[
\sum_{s=1}^{N} S_{j,s}^{(t)}
\le
s_{\ast}^{\,j}\sum_{s=1}^{N} C_{j,s,D},
\]
which proves the first inequality in \eqref{eq:sB_uniform_coeff_sparsity_clean}. The estimate \eqref{eq:CjsD_clean} is the standard stars-and-bars upper bound obtained by dropping the upper constraint \(\alpha_r\le D\). Summing that bound over \(s=1,\dots,N\) and using the hockey-stick identity gives
\[
\sum_{s=1}^{N}\binom{s+j-1}{j-1}
=
\binom{N+j}{j}-1,
\]
which proves the second inequality in \eqref{eq:sB_uniform_coeff_sparsity_clean}.

\medskip\noindent\textbf{Item~(2).}
Each row of a diagonal identity block contributes exactly one nonzero entry, while each row of a subdiagonal block \(-B(t)\) contains at most \(s_B\) nonzeros. Therefore a row of the full horizon matrix belongs either to the first block row, where it contains only the identity, or to a later block row, where it contains one diagonal identity entry together with one row from \(-B(t)\). Hence every row of \(M\) contains at most \(s_B+1\) nonzeros, proving \eqref{eq:sM_clean}.

\medskip\noindent\textbf{Item~(3).}
A row of \(M\in\mathbb{R}^{(T+1)\Delta_N\times (T+1)\Delta_N}\) is indexed by a pair
\[
(t,r),\qquad
t\in\{0,1,\dots,T\},\quad r\in\{1,2,\dots,\Delta_N\},
\]
where \(t\) is the block-row index and \(r\) is the in-block row index. If \(t=0\), the corresponding block row is the identity block, so row \((0,r)\) has exactly one nonzero entry at column \((0,r)\), and the rescaled row of \(\overline M\) still has exactly one nonzero entry. If \(t\ge 1\), then block row \(t\) contains the diagonal identity block in column block \(t\) and the strict subdiagonal block \(-B(t-1)\) in column block \(t-1\). Thus row \((t,r)\) consists of the diagonal entry at column \((t,r)\) together with the entries of row \(r\) of \(-B(t-1)\), so it has at most \(s_B+1\) nonzeros.

Given \((t,r)\), the optional row-sparse access oracle for \(\overline M\) therefore returns:
\begin{itemize}
\item if \(t=0\), the single diagonal entry in column \((0,r)\);
\item if \(t\ge 1\), the diagonal entry in column \((t,r)\) and the at most \(s_B\) nonzero entries of row \(r\) of \(-B(t-1)\) in column block \(t-1\).
\end{itemize}
The latter are obtained from the row-sparse access oracle for \(B(t-1)\). The final rescaling by \((1+\rho)^{-1}\) multiplies all returned values by a scalar and preserves sparsity. Thus one query to the row-sparse access oracle for \(\overline M\) uses only \(O(1)\) extra arithmetic on the block/time indices together with \(O(1)\) calls to the row-sparse access oracle of the active truncated lifted step matrix.
\end{proof}
Proposition~\ref{prop:sparsity_oracle_clean} therefore gives one row-sparse realization of the rescaled horizon matrix, although Theorem~\ref{thm:one_shot_folded_clean} still treats the final solver access model abstractly through hypothesis (H3).

\subsubsection{Truncation remainder and Carleman error}
\label{subsubsec:tail_truncation_clean}

The exact truncated lifted coordinates \(y^{(N)}(t)\) and the truncated lifted state \(\hat y(t)\)
coincide at time \(t=0\), but evolve differently thereafter because the truncated recurrence omits all
couplings from truncated levels into degrees above \(N\).
We now quantify that omission.

\begin{theorem}[Truncation error and cutoff design]
\label{thm:truncation_cutoff_clean}
Assume that the exact trajectory of the polynomial update model satisfies
\begin{equation}
\label{eq:vbar_clean}
\|v(t)\|_2\le \bar v<1,
\qquad t=0,1,\dots,T.
\end{equation}
For each \(j\in\{1,\dots,N\}\), define
\[
r_j(t):=\sum_{s=N+1}^{jD}K_{j,s}(t)\,y_s(t),
\qquad
r(t):=\big(r_1(t),r_2(t),\dots,r_N(t)\big)\in\mathcal H_N,
\]
and
\begin{equation}
\label{eq:GammaN_clean}
\Gamma_N
:=
\max_{0\le t\le T-1}
\left(
\sum_{j=1}^{N}
\left[
\sum_{s=N+1}^{jD}\|K_{j,s}(t)\|_2\,\bar v^{s}
\right]^2
\right)^{1/2}.
\end{equation}
Then the following hold.
\begin{enumerate}
\item \textbf{Exact truncated recurrence with tail forcing.}
The exact truncated lifted coordinates satisfy
\begin{equation}
\label{eq:exact_retained_recurrence_clean}
y^{(N)}(t+1)=B(t)\,y^{(N)}(t)+c(t)+r(t),
\qquad t=0,1,\dots,T-1,
\end{equation}
the truncation error
\[
\eta(t):=y^{(N)}(t)-\hat y(t)
\]
obeys
\begin{equation}
\label{eq:error_recurrence_clean}
\eta(t+1)=B(t)\eta(t)+r(t),
\qquad \eta(0)=0,
\end{equation}
and the tail forcing satisfies
\begin{equation}
\label{eq:tail_forcing_bound_clean}
\|r(t)\|_2\le \Gamma_N,
\qquad t=0,1,\dots,T-1.
\end{equation}

\item \textbf{Pointwise and stacked truncation error.}
Assume in addition \eqref{eq:contractive_B_clean}. Then for every \(t\in\{0,1,\dots,T\}\),
\begin{equation}
\label{eq:pointwise_truncation_bound_clean}
\|\eta(t)\|_2
\le
\frac{\Gamma_N}{1-\rho},
\end{equation}
and if
\[
Y_N^{\mathrm{ex}}:=\big(y^{(N)}(0),y^{(N)}(1),\dots,y^{(N)}(T)\big)\in\mathbb{R}^{N_h},
\qquad
\hat Y:=\big(\hat y(0),\hat y(1),\dots,\hat y(T)\big)\in\mathbb{R}^{N_h},
\]
then
\begin{equation}
\label{eq:horizon_truncation_bound_clean}
\|Y_N^{\mathrm{ex}}-\hat Y\|_2
\le
\sqrt{T+1}\,\frac{\Gamma_N}{1-\rho}.
\end{equation}

\item \textbf{Abstract cutoff design inequality.}
Let \(\varepsilon_{\mathrm{tr}}\in(0,1)\) be a target budget for the stacked truncation error in \eqref{eq:horizon_truncation_bound_clean}. It suffices to choose \(N\) so that
\begin{equation}
\label{eq:stacked_design_ineq_clean}
\sqrt{T+1}\,\frac{\Gamma_N}{1-\rho}
\le
\varepsilon_{\mathrm{tr}}.
\end{equation}
Equivalently, any explicit upper bound on \(\Gamma_N\) that makes \eqref{eq:stacked_design_ineq_clean} hold is sufficient for the main theorem.

\item \textbf{Explicit cutoff under a weighted coefficient sum.}
Fix \(\lambda>1\) and for each time index \(t\), define the weighted coefficient sum
\begin{equation}
\label{eq:weighted_coeff_sum_clean}
S_t(\lambda):=\sum_{\ell=0}^{D}\|Q_{\ell}^{(t)}\|_2\,(\lambda \bar v)^\ell.
\end{equation}
Assume additionally that
\begin{equation}
\label{eq:weighted_coeff_uniform_clean}
\sup_{0\le t\le T-1} S_t(\lambda)\le \chi<1.
\end{equation}
Then
\begin{equation}
\label{eq:GammaN_weighted_bound_clean}
\Gamma_N
\le
\lambda^{-(N+1)}\frac{\chi}{\sqrt{1-\chi^2}},
\end{equation}
and consequently a sufficient explicit cutoff choice is
\begin{equation}
\label{eq:weighted_cutoff_design_clean}
N\ge
\max\!\left\{
1,\;
\left\lceil
\frac{
\log\!\Big(
\frac{\sqrt{T+1}}{(1-\rho)\varepsilon_{\mathrm{tr}}}
\frac{\chi}{\sqrt{1-\chi^2}}
\Big)
}
{\log \lambda}
-1
\right\rceil
\right\},
\end{equation}
which guarantees \eqref{eq:stacked_design_ineq_clean}.
\end{enumerate}
\end{theorem}

\begin{proof}
For item~(1), Item~(1) of \Cref{prop:carleman_horizon_cond_clean} gives
\[
y_j(t+1)=\sum_{s=0}^{jD}K_{j,s}(t)y_s(t).
\]
Splitting the sum into the constant part \(s=0\), the retained levels \(1\le s\le N\), and the omitted tail \(s\ge N+1\) yields \eqref{eq:exact_retained_recurrence_clean}. Subtracting \eqref{eq:truncated_lifted_euler_clean} from this identity gives \eqref{eq:error_recurrence_clean}. Moreover,
\[
\|r_j(t)\|_2
\le
\sum_{s=N+1}^{jD}\|K_{j,s}(t)\|_2\,\|y_s(t)\|_2
=
\sum_{s=N+1}^{jD}\|K_{j,s}(t)\|_2\,\|v(t)\|_2^s
\le
\sum_{s=N+1}^{jD}\|K_{j,s}(t)\|_2\,\bar v^s,
\]
so, after squaring and summing over \(j=1,\dots,N\), one obtains \eqref{eq:tail_forcing_bound_clean}.

For item~(2), under \eqref{eq:contractive_B_clean}, iterating \eqref{eq:error_recurrence_clean} gives
\[
\eta(t)=\sum_{i=0}^{t-1}\bigl(B(t-1)\cdots B(i+1)\bigr)r(i),
\]
hence
\[
\|\eta(t)\|_2
\le
\sum_{i=0}^{t-1}\rho^{t-1-i}\|r(i)\|_2
\le
\frac{\Gamma_N}{1-\rho},
\]
which proves \eqref{eq:pointwise_truncation_bound_clean}. Summing this pointwise bound over \(t=0,\dots,T\) gives \eqref{eq:horizon_truncation_bound_clean}. Item~(3) is exactly the design inequality obtained by requiring \eqref{eq:horizon_truncation_bound_clean} to be at most \(\varepsilon_{\mathrm{tr}}\).

For item~(4), fix \(t\) and \(j\le N\). By definition of \(K_{j,s}(t)\),
\[
\sum_{s=N+1}^{jD}\|K_{j,s}(t)\|_2\bar v^s
\le
\sum_{s=N+1}^{jD}
\sum_{\substack{\alpha\in\{0,\dots,D\}^j\\|\alpha|=s}}
\prod_{r=1}^j\|Q_{\alpha_r}^{(t)}\|_2\bar v^{\alpha_r}.
\]
Since \(s\ge N+1\) and \(\lambda>1\), one has \(1\le \lambda^{s-(N+1)}\), so
\[
\begin{aligned}
\sum_{s=N+1}^{jD}\|K_{j,s}(t)\|_2\bar v^s
&\le
\lambda^{-(N+1)}
\sum_{s=N+1}^{jD}
\sum_{\substack{\alpha\in\{0,\dots,D\}^j\\|\alpha|=s}}
\prod_{r=1}^j\|Q_{\alpha_r}^{(t)}\|_2(\lambda\bar v)^{\alpha_r} \\
&\le
\lambda^{-(N+1)}
\sum_{\alpha\in\{0,\dots,D\}^j}
\prod_{r=1}^j\|Q_{\alpha_r}^{(t)}\|_2(\lambda\bar v)^{\alpha_r} \\
&=
\lambda^{-(N+1)}S_t(\lambda)^j.
\end{aligned}
\]
Using \eqref{eq:weighted_coeff_uniform_clean} and summing over \(j\) yields
\eqref{eq:GammaN_weighted_bound_clean}; substituting this into \eqref{eq:stacked_design_ineq_clean} gives \eqref{eq:weighted_cutoff_design_clean}.
\end{proof}
\subsection{QLSA complexity}
\begin{theorem}[QLSA in the input-state model]
\label{thm:qlsa_sparse}
Let $A\in\mathbb{R}^{N\times N}$ be invertible with $\|A\|_2\le 1$, and let $\kappa=\kappa_2(A)$. Assume that $A$ is supplied in the access model required by the chosen QLSA, with effective access parameter $s_A$, and let $C_A$ denote the gate cost of one query to that access model. Suppose the normalized right-hand side is supplied as an already prepared input state $|b\rangle$ proportional to $b\in\mathbb{R}^N$. Then there exists a quantum algorithm which, with success probability at least $2/3$, prepares a state $|\tilde x\rangle$ satisfying $\||\tilde x\rangle-|x\rangle\|_2\le \delta$, where $|x\rangle\propto A^{-1}b$, using
\[
\widetilde{\mathcal O}\!\left(s_A\,\kappa\,\polylog\!\frac{N}{\delta}\right)
\]
queries to the chosen solver access model and gate complexity
\[
\widetilde{\mathcal O}\!\left(s_A\,\kappa\,C_A\,\polylog\!\frac{N}{\delta}\right).
\]
\end{theorem}
This is the only linear-system routine used below. For Hermitian sparse systems, one concrete realization is the Childs--Kothari--Somma framework~\cite{childsQuantumAlgorithmSystems2017}; for general non-Hermitian systems, one common route is to pass through a Hermitian embedding or a block-encoding construction~\cite{gilyen2019qsvt}. Hypothesis (H3) states the solver model assumed for the horizon matrix, and $C_{\mathrm{prep}}(B_{\mathrm{rhs}})$ records the one-time cost of preparing the normalized right-hand side.

Under the assumptions of \Cref{thm:qlsa_sparse}, the total qubit cost to implement QLSA is upper bounded by
\[
\left\lceil \log_2 N\right\rceil
+
\mathcal{O}\!\left(\log\kappa+\log\frac{1}{\delta}+a_A\right),
\]
where the first term is the index/solution register and $a_A$ is the ancilla overhead required by the chosen
QLSA-access implementation of $A$.
Standard QLSA constructions in the input-state model use an $O(\log N)$-qubit system register, plus polylogarithmic workspace for phase estimation,
polynomial approximation, and success amplification; all are logarithmic in $\kappa$ and $1/\delta$ up to constant factors,
plus whatever ancillas are needed to implement the chosen QLSA-access routine for $A$.

\subsection{Quantum preparation of the training trajectory and final-state readout}

\begin{theorem}[Finite-step theorem for quantum trajectory preparation and final-state readout]
\label{thm:one_shot_folded_clean}
Consider the deterministic polynomial update model
\[
v(t+1)=\Psi_t(v(t)),
\qquad
\Psi_t(v)=\sum_{\ell=0}^{D}Q_{\ell}^{(t)}v^{\otimes \ell},
\]
defined by \eqref{eq:delta_folded_map}--\eqref{eq:full_step_vector_fields_clean}, with Carleman cutoff \(N\), lifted dimension
\(\Delta_N=\sum_{j=1}^{N}d^j\), and horizon dimension \(N_h=(T+1)\Delta_N\).
Let \(y^{(N)}(t)\) denote the exact truncated lifted coordinates of the trajectory generated by this model,
and let \(\hat y(t)\) denote the solution of the truncated lifted recurrence
\eqref{eq:truncated_lifted_euler_clean}.
Define the stacked vectors
\[
Y_N^{\mathrm{ex}}:=\big(y^{(N)}(0),y^{(N)}(1),\dots,y^{(N)}(T)\big),
\qquad
\hat Y:=\big(\hat y(0),\hat y(1),\dots,\hat y(T)\big).
\]
Let \(\varepsilon_{\mathrm{out}}>0\) be a target final output error budget.
Assume the following.
\begin{itemize}
\item[(H1)] \textbf{Contractive truncated lifted step.}
For the chosen cutoff \(N\), the truncated lifted step matrices satisfy
\[
\sup_{0\le t\le T-1}\|B(t)\|_2\le \rho<1.
\]
\item[(H2)] \textbf{Bounded trajectory in the local stability regime.}
The exact trajectory of the polynomial update model satisfies \(\|v(t)\|_2\le \bar v<1\) for all \(t=0,1,\dots,T\).
\item[(H3)] \textbf{Solver access model and right-hand-side preparation.}
Assume that each truncated lifted step matrix \(B(t)\) is \(s_B\)-row sparse, and that the rescaled horizon matrix
\[
\overline M:=\frac{1}{1+\rho}M
\]
is supplied in the solver access model required by Theorem~\ref{thm:qlsa_sparse}, with effective access parameter \(s_M\) satisfying
\[
s_M\le s_B+1.
\]
Let \(C_{\mathrm{SA}}\) denote the gate cost of one query to that chosen QLSA access model.
Assume further that there is an input-state preparation routine which prepares the normalized right-hand-side state
\[
|B_{\mathrm{rhs}}\rangle:=\frac{B_{\mathrm{rhs}}}{\|B_{\mathrm{rhs}}\|_2}
\]
once, with gate cost \(C_{\mathrm{prep}}(B_{\mathrm{rhs}})\), and then supplies it to the QLSA as the right-hand-side input register.
\item[(H4)] \textbf{Quantitative initialization lower bound.}
The initial truncated lifted vector satisfies
\[
\|y^{(N)}(0)\|_2 \ge \beta_0 > 0.
\]
\item[(H5)] \textbf{Terminal degree-1 parameter-block weight.}
Let
\[
\Pi_{\mathrm{term},u}
\]
denote the orthogonal projector from the stacked horizon space onto the time-$T$, degree-1 parameter block, namely the last \(n\) coordinates of the level-1 block at terminal time \(T\) because \(v=(\delta,u)\).
Assume the normalized exact stacked trajectory satisfies the terminal-weight lower bound
\begin{equation}
\label{eq:terminal_weight_clean}
p_{\mathrm{term}}
:=
\|\Pi_{\mathrm{term},u}|Y_N^{\mathrm{ex}}\rangle\|_2^2
\ge p_* >0.
\end{equation}
\item[(H6)] \textbf{Sparse-output tomography/readout model.}
Let
\[
|u_T^{\mathrm{ex}}\rangle
:=
\frac{\Pi_{\mathrm{term},u}|Y_N^{\mathrm{ex}}\rangle}{\|\Pi_{\mathrm{term},u}|Y_N^{\mathrm{ex}}\rangle\|_2}
\]
denote the normalized terminal degree-1 parameter state.
Assume there is a readout routine which, using copies of an approximation to \(|u_T^{\mathrm{ex}}\rangle\), returns a classical description of that final parameter state to error at most \(\varepsilon_{\mathrm{ro}}\) with additional cost
\[
C_{\mathrm{ro}}(n,s_{\mathrm{out}},\varepsilon_{\mathrm{ro}})
\]
under an output-sparsity bound \(s_{\mathrm{out}}\).
In the sparse-output regime of interest, this cost is assumed to be efficient, for example poly\((s_{\mathrm{out}},\log n,1/\varepsilon_{\mathrm{ro}})\), as in structured and compressed tomography settings \cite{cramerEfficientQuantumState2010,grossQuantumStateTomographyCompressed2010,gutaFastStateTomography2020}.
\end{itemize}
Let \(\Gamma_N\) be the tail constant from \eqref{eq:GammaN_clean}. Then the following hold.

\begin{enumerate}
\item \textbf{Truncation error on the stacked trajectory.}
\begin{equation}
\label{eq:main_truncation_clean}
\|Y_N^{\mathrm{ex}}-\hat Y\|_2
\le
\varepsilon_{\mathrm{tr}}^{\mathrm{hor}}(N,T)
:=
\sqrt{T+1}\,\frac{\Gamma_N}{1-\rho}.
\end{equation}

\item \textbf{Condition number and structural row sparsity of the horizon matrix.}
The horizon matrix \(M\) in \eqref{eq:horizon_system_clean} satisfies
\begin{equation}
\label{eq:main_kappa_sparsity_clean}
\kappa_2(M)
\le
\min\!\left\{\frac{1+\rho}{1-\rho},\,2(T+1)\right\},
\qquad
s_{\mathrm{row}}(M)\le s_B+1.
\end{equation}

\item \textbf{QLSA query and gate complexity.}
Let
\[
\overline M:=\frac{1}{1+\rho}M.
\]
Then \(\|\overline M\|_2\le 1\), \(\kappa_2(\overline M)=\kappa_2(M)\), and hypothesis (H3) supplies the solver access model used by \Cref{thm:qlsa_sparse} with effective parameter \(s_M\) and per-query gate cost \(C_{\mathrm{SA}}\).
With success probability at least \(2/3\), a QLSA applied to
\[
\overline M Y=(1+\rho)^{-1}B_{\mathrm{rhs}}
\]
prepares a state \(|\widetilde Y\rangle\) satisfying
\begin{equation}
\label{eq:qlsa_state_error_clean}
\big\|\,|\widetilde Y\rangle-|\hat Y\rangle\,\big\|_2
\le
\varepsilon_{\mathrm{LS}}.
\end{equation}
Its query complexity is
\begin{equation}
\label{eq:one_shot_query_complexity_clean}
\widetilde{\mathcal O}\!\left(
s_M\,\kappa_2(M)\,
\operatorname{polylog}\!\frac{N_h}{\varepsilon_{\mathrm{LS}}}
\right),
\end{equation}
and its gate complexity is
\begin{equation}
\label{eq:one_shot_gate_complexity_clean}
\widetilde{\mathcal O}\!\left(
C_{\mathrm{prep}}(B_{\mathrm{rhs}})
+
s_M\,\kappa_2(M)\,C_{\mathrm{SA}}\,
\operatorname{polylog}\!\frac{N_h}{\varepsilon_{\mathrm{LS}}}
\right).
\end{equation}
Under the displayed hypothesis \(s_M\le s_B+1\), these bounds are in particular upper bounded by the same expressions with \(s_B+1\) in place of \(s_M\).

\item \textbf{End-to-end state error with respect to the exact truncated lifted trajectory.}
Conditioned on the successful QLSA preparation event from item~(3), the normalized quantum encoding produced by the QLSA obeys
\begin{equation}
\label{eq:one_shot_end_to_end_state_clean}
\big\|\,|\widetilde Y\rangle-|Y_N^{\mathrm{ex}}\rangle\,\big\|_2
\le
\varepsilon_{\mathrm{LS}}
+
\frac{2\,\varepsilon_{\mathrm{tr}}^{\mathrm{hor}}(N,T)}{\beta_0}.
\end{equation}
Consequently, if
\begin{equation}
\label{eq:one_shot_budget_split_clean}
\varepsilon_{\mathrm{LS}}\le \frac{\varepsilon_{\mathrm{out}}}{2},
\qquad
\varepsilon_{\mathrm{tr}}^{\mathrm{hor}}(N,T)\le \frac{\beta_0\,\varepsilon_{\mathrm{out}}}{4},
\end{equation}
then
\begin{equation}
\label{eq:one_shot_final_eps_clean}
\big\|\,|\widetilde Y\rangle-|Y_N^{\mathrm{ex}}\rangle\,\big\|_2
\le
\varepsilon_{\mathrm{out}}.
\end{equation}

\item \textbf{Terminal degree-1 parameter-block extraction and sparse-output readout.}
Assume, in addition to item~(4), that
\begin{equation}
\label{eq:terminal_nonzero_condition_clean}
\varepsilon_{\mathrm{LS}}
+
\frac{2\,\varepsilon_{\mathrm{tr}}^{\mathrm{hor}}(N,T)}{\beta_0}
\le \frac{\sqrt{p_*}}{2}.
\end{equation}
Then one extra ancilla can mark the known terminal degree-1 parameter block, and the corresponding ancilla-based preparation protocol isolates that normalized terminal state with constant success probability.
The resulting terminal state \(|\widetilde u_T\rangle\) satisfies
\begin{equation}
\label{eq:terminal_block_state_error_clean}
\bigl\|\,|\widetilde u_T\rangle-|u_T^{\mathrm{ex}}\rangle\,\bigr\|_2
\le
\frac{2}{\sqrt{p_*}}
\left(
\varepsilon_{\mathrm{LS}}
+
\frac{2\,\varepsilon_{\mathrm{tr}}^{\mathrm{hor}}(N,T)}{\beta_0}
\right).
\end{equation}
Applying the sparse-output readout routine from \emph{(H6)} with readout error \(\varepsilon_{\mathrm{ro}}\) yields a classical description of the final parameter state with total error at most
\begin{equation}
\label{eq:terminal_readout_total_error_clean}
\frac{2}{\sqrt{p_*}}
\left(
\varepsilon_{\mathrm{LS}}
+
\frac{2\,\varepsilon_{\mathrm{tr}}^{\mathrm{hor}}(N,T)}{\beta_0}
\right)
+
\varepsilon_{\mathrm{ro}},
\end{equation}
and with additional overhead
\begin{equation}
\label{eq:terminal_readout_total_cost_clean}
C_{\mathrm{ro}}(n,s_{\mathrm{out}},\varepsilon_{\mathrm{ro}}).
\end{equation}
Consequently, if
\[
\varepsilon_{\mathrm{LS}}
+
\frac{2\,\varepsilon_{\mathrm{tr}}^{\mathrm{hor}}(N,T)}{\beta_0}
\le
\frac{\sqrt{p_*}\,\varepsilon_{\mathrm{out}}}{4},
\qquad
\varepsilon_{\mathrm{ro}}\le \frac{\varepsilon_{\mathrm{out}}}{2},
\]
then the final classical output error is at most \(\varepsilon_{\mathrm{out}}\).
If, furthermore,
\[
C_{\mathrm{ro}}(n,s_{\mathrm{out}},\varepsilon_{\mathrm{ro}})
=
\widetilde{\mathcal O}\!\bigl(\poly(s_{\mathrm{out}},\log n,1/\varepsilon_{\mathrm{ro}})\bigr),
\]
then the final readout stage is efficient and does not change the theorem's dominant polylogarithmic dependence on model size.
\end{enumerate}
\end{theorem}

\begin{proof}
Item~(1) is Item~(2) of \Cref{thm:truncation_cutoff_clean}. Item~(2) combines Item~(3) of \Cref{prop:carleman_horizon_cond_clean} with Item~(2) of \Cref{prop:sparsity_oracle_clean}.

For item~(3), Item~(3) of \Cref{prop:carleman_horizon_cond_clean} gives \(\|M\|_2\le 1+\rho\). Hence for \(\overline M=(1+\rho)^{-1}M\) one has \(\|\overline M\|_2\le 1\) and \(\kappa_2(\overline M)=\kappa_2(M)\). Hypothesis (H3) supplies the access model required by \Cref{thm:qlsa_sparse}, together with the normalized right-hand-side state \(|B_{\mathrm{rhs}}\rangle\). The normalized right-hand side is unchanged by the scalar factor \((1+\rho)^{-1}\). Applying \Cref{thm:qlsa_sparse} to
\[
\overline M Y=(1+\rho)^{-1}B_{\mathrm{rhs}}
\]
yields, with success probability at least \(2/3\), a state \(|\widetilde Y\rangle\) satisfying \eqref{eq:qlsa_state_error_clean} and the query/gate bounds \eqref{eq:one_shot_query_complexity_clean}--\eqref{eq:one_shot_gate_complexity_clean}.

For item~(4), the first block of \(Y_N^{\mathrm{ex}}\) is \(y^{(N)}(0)\), so
\begin{equation}
\label{eq:Ynorm_lower_clean}
\|Y_N^{\mathrm{ex}}\|_2\ge \|y^{(N)}(0)\|_2\ge \beta_0.
\end{equation}
Using
\[
\left\|\frac{a}{\|a\|_2}-\frac{b}{\|b\|_2}\right\|_2
\le
\frac{2\|a-b\|_2}{\|a\|_2},
\]
with \(a=Y_N^{\mathrm{ex}}\) and \(b=\hat Y\), and then applying item~(1), gives
\[
\bigl\|\,|Y_N^{\mathrm{ex}}\rangle-|\hat Y\rangle\,\bigr\|_2
\le
\frac{2\,\varepsilon_{\mathrm{tr}}^{\mathrm{hor}}(N,T)}{\beta_0}.
\]
Conditioning on the successful QLSA preparation event from item~(3), one also has
\(\bigl\|\,|\widetilde Y\rangle-|\hat Y\rangle\,\bigr\|_2\le \varepsilon_{\mathrm{LS}}\).
The triangle inequality therefore yields \eqref{eq:one_shot_end_to_end_state_clean}, and the budget split \eqref{eq:one_shot_budget_split_clean} implies \eqref{eq:one_shot_final_eps_clean}.

For item~(5), set
\[
\varepsilon_{\mathrm{state}}
:=
\bigl\|\,|\widetilde Y\rangle-|Y_N^{\mathrm{ex}}\rangle\,\bigr\|_2.
\]
By item~(4),
\[
\varepsilon_{\mathrm{state}}
\le
\varepsilon_{\mathrm{LS}}
+
\frac{2\,\varepsilon_{\mathrm{tr}}^{\mathrm{hor}}(N,T)}{\beta_0}.
\]
Since \(\Pi_{\mathrm{term},u}\) projects onto a known subset of the computational basis labels, a reversible indicator circuit with one ancilla marks exactly that sector.
Orthogonal projection is contractive, so
\[
\|\Pi_{\mathrm{term},u}|\widetilde Y\rangle - \Pi_{\mathrm{term},u}|Y_N^{\mathrm{ex}}\rangle\|_2
\le
\varepsilon_{\mathrm{state}}.
\]
Using \eqref{eq:terminal_weight_clean} and \eqref{eq:terminal_nonzero_condition_clean},
\[
\|\Pi_{\mathrm{term},u}|\widetilde Y\rangle\|_2
\ge
\|\Pi_{\mathrm{term},u}|Y_N^{\mathrm{ex}}\rangle\|_2-\varepsilon_{\mathrm{state}}
\ge
\sqrt{p_*}-\varepsilon_{\mathrm{state}}
\ge
\frac{\sqrt{p_*}}{2}.
\]
The ancilla-based preparation protocol for this known marked sector therefore produces the normalized terminal state with constant success probability.
Applying
\[
\left\|
\frac{a}{\|a\|_2}-\frac{b}{\|b\|_2}
\right\|_2
\le
\frac{2\|a-b\|_2}{\|a\|_2}
\]
to
\[
a:=\Pi_{\mathrm{term},u}|Y_N^{\mathrm{ex}}\rangle,
\qquad
b:=\Pi_{\mathrm{term},u}|\widetilde Y\rangle,
\]
and using \(\|a\|_2\ge \sqrt{p_*}\), proves \eqref{eq:terminal_block_state_error_clean}.
The total readout error \eqref{eq:terminal_readout_total_error_clean} and cost \eqref{eq:terminal_readout_total_cost_clean} then follow from \emph{(H6)} by the triangle inequality together with the constant-success-probability terminal-state preparation step.
\end{proof}

\begin{corollary}[From the main theorem to the exact adversarial-training dynamics]
\label{cor:exact_dynamics_bridge_clean}
Assume the hypotheses of Theorem~\ref{thm:one_shot_folded_clean}, with the same contractive matrices \(B(t)\) and the same truncated lifted recurrence used to define \(\hat Y\).
Assume in addition that the truncated lifted coordinates of the exact adversarial-training dynamics satisfy
\begin{equation}
\label{eq:exact_dynamics_forcing_decomposition_clean}
y_{\mathrm{phys}}(t+1)
=
B(t)\,y_{\mathrm{phys}}(t)+c(t)+e_{\mathrm{tr}}(t)+e_{\mathrm{model}}(t),
\qquad t=0,1,\dots,T-1,
\end{equation}
with the same initialization
\[
y_{\mathrm{phys}}(0)=y^{(N)}(0).
\]
Assume the truncation error term satisfies
\begin{equation}
\label{eq:bridge_truncation_forcing_budget_clean}
\|e_{\mathrm{tr}}(t)\|_2\le \Gamma_N
\qquad
\text{for all }t=0,1,\dots,T-1,
\end{equation}
where \(\Gamma_N\) is defined in \eqref{eq:GammaN_clean}, and assume the modeling error term satisfies
\begin{equation}
\label{eq:bridge_model_step_budget_clean}
\|e_{\mathrm{model}}(t)\|_2
\le
L_{N,\bar v}^{\mathrm{lift}}\,\varepsilon_{\mathrm{base,step}}
\qquad
\text{for all }t=0,1,\dots,T-1,
\end{equation}
where \(L_{N,\bar v}^{\mathrm{lift}}\) is defined in \eqref{eq:lift_local_lipschitz_constant_clean}.
Define the combined window error bound for the exact dynamics
\begin{equation}
\label{eq:phys_horizon_error_budget_clean}
\varepsilon_{\mathrm{phys}}^{\mathrm{hor}}(N,T)
:=
\frac{\sqrt{T+1}}{1-\rho}
\left(
\Gamma_N
+
L_{N,\bar v}^{\mathrm{lift}}\,\varepsilon_{\mathrm{base,step}}
\right).
\end{equation}
Let \(\varepsilon_{\mathrm{out}}>0\) be a target normalized-state error for the exact trajectory state.
Conditioned on the successful QLSA preparation from Theorem~\ref{thm:one_shot_folded_clean}, the QLSA output obeys
\begin{equation}
\label{eq:exact_dynamics_state_error_clean}
\bigl\|\,|\widetilde Y\rangle-|Y_{\mathrm{phys}}\rangle\,\bigr\|_2
\le
\varepsilon_{\mathrm{LS}}
+
\frac{2\,\varepsilon_{\mathrm{phys}}^{\mathrm{hor}}(N,T)}{\beta_0}.
\end{equation}
Consequently, if
\begin{equation}
\label{eq:exact_dynamics_budget_split_clean}
\varepsilon_{\mathrm{LS}}\le \frac{\varepsilon_{\mathrm{out}}}{2},
\qquad
\varepsilon_{\mathrm{phys}}^{\mathrm{hor}}(N,T)\le \frac{\beta_0\,\varepsilon_{\mathrm{out}}}{4},
\end{equation}
then
\begin{equation}
\label{eq:exact_dynamics_final_eps_clean}
\bigl\|\,|\widetilde Y\rangle-|Y_{\mathrm{phys}}\rangle\,\bigr\|_2
\le \varepsilon_{\mathrm{out}}.
\end{equation}
If, in addition, the physical terminal degree-1 parameter block satisfies the same lower bound
\begin{equation}
\label{eq:terminal_weight_phys_clean}
\|\Pi_{\mathrm{term},u}|Y_{\mathrm{phys}}\rangle\|_2^2\ge p_*>0,
\end{equation}
and the sparse-output readout model \emph{(H6)} is assumed for the normalized physical terminal state
\[
|u_{\mathrm{phys}}(T)\rangle
:=
\frac{\Pi_{\mathrm{term},u}|Y_{\mathrm{phys}}\rangle}{\|\Pi_{\mathrm{term},u}|Y_{\mathrm{phys}}\rangle\|_2},
\]
then, whenever
\begin{equation}
\label{eq:terminal_nonzero_condition_phys_clean}
\varepsilon_{\mathrm{LS}}
+
\frac{2\,\varepsilon_{\mathrm{phys}}^{\mathrm{hor}}(N,T)}{\beta_0}
\le \frac{\sqrt{p_*}}{2},
\end{equation}
one extra ancilla marking the terminal sector together with the same ancilla-based preparation protocol yields, with constant success probability, a terminal state \(|\widetilde u_T\rangle\) with
\begin{equation}
\label{eq:terminal_block_state_error_phys_clean}
\bigl\|\,|\widetilde u_T\rangle-|u_{\mathrm{phys}}(T)\rangle\,\bigr\|_2
\le
\frac{2}{\sqrt{p_*}}
\left(
\varepsilon_{\mathrm{LS}}
+
\frac{2\,\varepsilon_{\mathrm{phys}}^{\mathrm{hor}}(N,T)}{\beta_0}
\right),
\end{equation}
and the resulting classical final-output error is at most
\begin{equation}
\label{eq:terminal_readout_total_error_phys_clean}
\frac{2}{\sqrt{p_*}}
\left(
\varepsilon_{\mathrm{LS}}
+
\frac{2\,\varepsilon_{\mathrm{phys}}^{\mathrm{hor}}(N,T)}{\beta_0}
\right)
+
\varepsilon_{\mathrm{ro}},
\end{equation}
with additional overhead
\begin{equation}
\label{eq:terminal_readout_total_cost_phys_clean}
C_{\mathrm{ro}}(n,s_{\mathrm{out}},\varepsilon_{\mathrm{ro}}).
\end{equation}
In particular, this final classical output error is at most \(\varepsilon_{\mathrm{out}}\) whenever
\[
\varepsilon_{\mathrm{LS}}
+
\frac{2\,\varepsilon_{\mathrm{phys}}^{\mathrm{hor}}(N,T)}{\beta_0}
\le
\frac{\sqrt{p_*}\,\varepsilon_{\mathrm{out}}}{4},
\qquad
\varepsilon_{\mathrm{ro}}\le \frac{\varepsilon_{\mathrm{out}}}{2}.
\]
\end{corollary}

\begin{proof}
Set
\[
z(t):=\hat y(t),
\qquad
\tilde z(t):=y_{\mathrm{phys}}(t),
\qquad
e(t):=e_{\mathrm{tr}}(t)+e_{\mathrm{model}}(t).
\]
Then \eqref{eq:exact_dynamics_forcing_decomposition_clean} and the truncated lifted recurrence give
\[
z(t+1)=B(t)z(t)+c(t),
\qquad
\tilde z(t+1)=B(t)\tilde z(t)+c(t)+e(t),
\qquad
z(0)=\tilde z(0)=y^{(N)}(0).
\]
By \eqref{eq:bridge_truncation_forcing_budget_clean} and \eqref{eq:bridge_model_step_budget_clean},
\[
\|e(t)\|_2
\le
\Gamma_N+L_{N,\bar v}^{\mathrm{lift}}\,\varepsilon_{\mathrm{base,step}}
\qquad\text{for all }t.
\]
Using the standard contractive-forcing estimate under \(\sup_t\|B(t)\|_2\le \rho<1\),
\[
y_{\mathrm{phys}}(t)-\hat y(t)
=
\sum_{i=0}^{t-1}
\Big(B(t-1)B(t-2)\cdots B(i+1)\Big)e(i),
\]
so
\[
\|y_{\mathrm{phys}}(t)-\hat y(t)\|_2
\le
\sum_{i=0}^{t-1}\rho^{t-1-i}\|e(i)\|_2
\le
\frac{\Gamma_N+L_{N,\bar v}^{\mathrm{lift}}\,\varepsilon_{\mathrm{base,step}}}{1-\rho}.
\]
Summing this pointwise bound over the full training window yields
\[
\|Y_{\mathrm{phys}}-\hat Y\|_2
\le
\sqrt{T+1}\,
\frac{\Gamma_N+L_{N,\bar v}^{\mathrm{lift}}\,\varepsilon_{\mathrm{base,step}}}{1-\rho}
=
\varepsilon_{\mathrm{phys}}^{\mathrm{hor}}(N,T).
\]
Because the first block of \(Y_{\mathrm{phys}}\) is \(y^{(N)}(0)\), assumption (H4) gives
\[
\|Y_{\mathrm{phys}}\|_2\ge \|y^{(N)}(0)\|_2\ge \beta_0.
\]
Using
\[
\left\|
\frac{a}{\|a\|_2}-\frac{b}{\|b\|_2}
\right\|_2
\le
\frac{2\|a-b\|_2}{\|a\|_2},
\]
with \(a=Y_{\mathrm{phys}}\) and \(b=\hat Y\), yields
\[
\bigl\|\,|Y_{\mathrm{phys}}\rangle-|\hat Y\rangle\,\bigr\|_2
\le
\frac{2\,\varepsilon_{\mathrm{phys}}^{\mathrm{hor}}(N,T)}{\beta_0}.
\]
Combining this with
\(\bigl\|\,|\widetilde Y\rangle-|\hat Y\rangle\,\bigr\|_2\le \varepsilon_{\mathrm{LS}}\)
on the successful QLSA event proves \eqref{eq:exact_dynamics_state_error_clean}, and \eqref{eq:exact_dynamics_budget_split_clean} implies \eqref{eq:exact_dynamics_final_eps_clean}.

The terminal-readout part is the same projection-and-normalization argument as in item~(5) of Theorem~\ref{thm:one_shot_folded_clean}, now applied to the pair \((|\widetilde Y\rangle,|Y_{\mathrm{phys}}\rangle)\) and the state-error bound \eqref{eq:exact_dynamics_state_error_clean}. Using \eqref{eq:terminal_weight_phys_clean} and \eqref{eq:terminal_nonzero_condition_phys_clean} gives the constant-success-probability terminal-state preparation step, the normalized terminal-state bound \eqref{eq:terminal_block_state_error_phys_clean}, and then \eqref{eq:terminal_readout_total_error_phys_clean}--\eqref{eq:terminal_readout_total_cost_phys_clean} by the triangle inequality and the assumed readout routine from \emph{(H6)}.
\end{proof}
The $T+1$ factor in Item~(2) of \Cref{thm:truncation_cutoff_clean}, Theorem~\ref{thm:one_shot_folded_clean}, and Corollary~\ref{cor:exact_dynamics_bridge_clean} appears because the QLSA prepares the full stacked trajectory encoding across the entire window. When only a terminal output is required, item~(5) of \Cref{thm:one_shot_folded_clean} and the final paragraph of \Cref{cor:exact_dynamics_bridge_clean} isolate the additional extraction and readout overhead from the core horizon-state preparation cost.

\section{Appendix discussion}

The main theorem rewrites a fixed training window of projected-gradient robust training as a sparse linear-system problem for the polynomial update model and, under terminal-weight and sparse-output assumptions, turns that trajectory encoding into an efficiently readable final parameter state.

For classical simulation, the main burden lies in the lifted dimension \(\Delta_N=\sum_{j=1}^N d^j\) and the horizon dimension \(N_h=(T+1)\Delta_N\). In the QLSA formulation, by contrast, the dependence on \(N_h\) is only polylogarithmic, and the solver register uses \(O(\log N_h)\) qubits up to ancilla overhead.

The transfer to the exact robust-training dynamics is given in \Cref{cor:exact_dynamics_bridge_clean}. The numerical experiment in the main text verifies the classical implementation of the reduced update model analyzed in the theorem in one concrete setting.

\section{Numerical implementation details for the main-text experiment}

The numerical experiment in the main text verifies the classical implementation of the reduced update model analyzed in the theorem, using one concrete trajectory on the reduced MNIST task.

We restrict MNIST to digits $0$--$4$. Each image is resized to $12\times 12$, flattened to $144$ features, and mapped by a fixed random projection to $10$ dimensions. The classifier is a bias-free $10\to 4\to 5$ model with $60$ trainable parameters. We consider three training modes: clean-only, robust-only, and mixed training with
\[
\Loss_{\mathrm{mix}}=(1-\alpha)\Loss_{\mathrm{clean}}+\alpha\Loss_{\mathrm{rob}},
\qquad \alpha=0.50.
\]
Here \(\Loss_{\mathrm{clean}}=\Loss(f_u(x),y)\), while \(\Loss_{\mathrm{rob}}\) is used only inside the reduced classical simulation that implements the reduced update model analyzed in the main text. The reported figure uses $1.2\times 10^5$ alternating adversarial-training steps with batch size $5$.

Robust evaluation uses a PGD attack with $\epsilon=0.025$, step size $0.01$, and $10$ attack steps. The nine panels in the main-text figure report robust accuracy, clean accuracy, and clean loss for the three training modes. For this experiment, the internal linear-system component is implemented with the MindQuantum HHL QSLP backend. The plotted trajectories are used to confirm that the reduced update model is implemented consistently over the full training window and produces the expected behavior in the three training modes.

\section{Extensions}

Deterministic-$(K_t,L_t)$ schedules can be handled as follows.
\begin{proposition}[Deterministic-$(K_t,L_t)$ outer-update composition]
\label{prop:fixed_KL_composition_clean}
Fix deterministic schedules \(\{K_t\}_{t=0}^{T-1}\) and \(\{L_t\}_{t=0}^{T-1}\) satisfying
\[
1\le K_t\le K_{\max},
\qquad
1\le L_t\le L_{\max}
\qquad
\text{for all }t.
\]
For each outer time \(t\), let the exact attack substeps and learner substeps be
\[
A_{t,j}^{\mathrm{phys}}(\delta,u)
:=
\begin{pmatrix}
\operatorname{clip}\!\Big(
\delta+\eta_{\delta,t,j}\,\operatorname{sign}\big(g_{\delta,t,j}(\delta,u)\big),
 -\epsilon,+\epsilon
\Big)\\
u
\end{pmatrix},
\qquad j=0,1,\dots,K_t-1,
\]
\[
U_{t,\ell}^{\mathrm{phys}}(\delta,u)
:=
\begin{pmatrix}
\delta\\
u-\eta_{u,t,\ell}\,g_{u,t,\ell}(\delta,u)
\end{pmatrix},
\qquad \ell=0,1,\dots,L_t-1,
\]
where the attack substeps act with the learner frozen and the learner substeps act with the perturbation frozen. Let the corresponding polynomial attack and learner substeps be
\[
A_{t,j}(\delta,u)
:=
\begin{pmatrix}
\epsilon\,P_{\mathrm c}\!\Big(
\big[\delta+\eta_{\delta,t,j}P_{\mathrm s}\big(\mathcal G_{\delta,t,j}(\delta,u)/\alpha_{t,j}\big)\big]/\epsilon
\Big)\\
u
\end{pmatrix},
\]
\[
U_{t,\ell}(\delta,u)
:=
\begin{pmatrix}
\delta\\
u-\eta_{u,t,\ell}\,\mathcal G_{u,t,\ell}(\delta,u)
\end{pmatrix},
\]
with deterministic substep schedules encoded in the indices \(t,j,\ell\). Define the exact and polynomial outer maps
\[
\Psi_t^{\mathrm{phys},K_t,L_t}
:=
U_{t,L_t-1}^{\mathrm{phys}}\circ\cdots\circ U_{t,0}^{\mathrm{phys}}
\circ
A_{t,K_t-1}^{\mathrm{phys}}\circ\cdots\circ A_{t,0}^{\mathrm{phys}},
\]
\[
\Psi_t^{K_t,L_t}
:=
U_{t,L_t-1}\circ\cdots\circ U_{t,0}
\circ
A_{t,K_t-1}\circ\cdots\circ A_{t,0}.
\]
Then the following hold.
\begin{enumerate}
\item \textbf{Polynomiality and degree growth under finite composition.}
Assume each polynomial gradient approximation \(\mathcal G_{\delta,t,j}\) and \(\mathcal G_{u,t,\ell}\) has degree at most \(q\). Then each attack substep \(A_{t,j}\) has degree at most
\[
D_A\le qK_sK_c,
\]
and each learner substep \(U_{t,\ell}\) has degree at most \(q\). Consequently the composed outer map \(\Psi_t^{K_t,L_t}\) is polynomial with degree at most
\begin{equation}
\label{eq:fixed_KL_degree_bound_clean}
D_t
\le
D_A^{K_t} q^{L_t}
\le
q^{K_{\max}+L_{\max}}(K_sK_c)^{K_{\max}}
=:
D_{\mathrm{sched}}.
\end{equation}
Hence, after padding missing higher-degree coefficients by zero, one may write
\[
\Psi_t^{K_t,L_t}(v)=\sum_{\ell=0}^{D_{\mathrm{sched}}}Q_{\ell}^{(t)}\,v^{\otimes \ell}.
\]

\item \textbf{One-step state error for the composed outer update.}
Let \(\mathcal V\subseteq\mathbb R^d\) be a local domain that is forward invariant under every exact and polynomial substep above. Assume there exist constants \(\varepsilon_{A,\mathrm{sub}}\ge 0\), \(\varepsilon_{U,\mathrm{sub}}\ge 0\), and \(\Lambda\ge 0\) such that uniformly for all admissible \(t,j,\ell\),
\[
\sup_{v\in\mathcal V}
\|A_{t,j}^{\mathrm{phys}}(v)-A_{t,j}(v)\|_2
\le
\varepsilon_{A,\mathrm{sub}},
\qquad
\sup_{v\in\mathcal V}
\|U_{t,\ell}^{\mathrm{phys}}(v)-U_{t,\ell}(v)\|_2
\le
\varepsilon_{U,\mathrm{sub}},
\]
and
\[
\operatorname{Lip}(A_{t,j}^{\mathrm{phys}}|_{\mathcal V}),
\ \operatorname{Lip}(A_{t,j}|_{\mathcal V}),
\ \operatorname{Lip}(U_{t,\ell}^{\mathrm{phys}}|_{\mathcal V}),
\ \operatorname{Lip}(U_{t,\ell}|_{\mathcal V})
\le
\Lambda.
\]
Set
\[
M_t:=K_t+L_t,
\qquad
\varepsilon_{\mathrm{sub}}
:=
\max\{\varepsilon_{A,\mathrm{sub}},\varepsilon_{U,\mathrm{sub}}\},
\qquad
C_t(\Lambda)
:=
\sum_{r=0}^{M_t-1}\Lambda^r.
\]
Define the uniform schedule constant
\[
C_{\mathrm{sched}}(\Lambda)
:=
\max_{0\le \tau\le T-1} C_{\tau}(\Lambda)
=
\max_{0\le \tau\le T-1}\sum_{r=0}^{K_{\tau}+L_{\tau}-1}\Lambda^r.
\]
Then every \(v\in\mathcal V\) satisfies
\begin{equation}
\label{eq:fixed_KL_base_step_bound_clean}
\|\Psi_t^{\mathrm{phys},K_t,L_t}(v)-\Psi_t^{K_t,L_t}(v)\|_2
\le
C_t(\Lambda)\,\varepsilon_{\mathrm{sub}}
\le
C_{\mathrm{sched}}(\Lambda)\,\varepsilon_{\mathrm{sub}}.
\end{equation}
In particular, since \(M_t\le K_{\max}+L_{\max}\) and \(\Lambda\ge 0\),
\[
C_{\mathrm{sched}}(\Lambda)
\le
\sum_{r=0}^{K_{\max}+L_{\max}-1}\Lambda^r
=
\begin{cases}
K_{\max}+L_{\max}, & \Lambda=1,\\[1ex]
\dfrac{\Lambda^{K_{\max}+L_{\max}}-1}{\Lambda-1}, & \Lambda\neq 1.
\end{cases}
\]
In particular, \(\varepsilon_{A,\mathrm{sub}}\) may be supplied by the one-attack-substep sign/clip approximation bound of \Cref{prop:sign_clip_design_clean}, while \(\varepsilon_{U,\mathrm{sub}}\) may be bounded by the learner-gradient approximation budget for one learner substep.

\item \textbf{Application of Theorem~\ref{thm:one_shot_folded_clean} to deterministic-$(K_t,L_t)$ schedules.}
If the composed polynomial outer maps \(\Psi_t^{K_t,L_t}\) satisfy the coefficient-tail, sparsity, contractivity, and initialization assumptions used in \Cref{thm:one_shot_folded_clean}, then that theorem applies unchanged with \(D\) replaced by \(D_{\mathrm{sched}}\). Moreover, item~(2) implies \eqref{eq:base_step_bridge_assumption_clean} with
\begin{equation}
\label{eq:fixed_KL_eps_base_clean}
\varepsilon_{\mathrm{base,step}}
:=
C_{\mathrm{sched}}(\Lambda)\,\varepsilon_{\mathrm{sub}},
\end{equation}
so item~(2) of \Cref{prop:base_to_lifted_bridge_clean} and Corollary~\ref{cor:exact_dynamics_bridge_clean} also extend unchanged to any deterministic schedules \(K_t\) and \(L_t\) with finite uniform bounds.
\end{enumerate}
\end{proposition}

\begin{proof}
\textbf{Item~(1).}
The attack-substep degree bound is exactly the same calculation as in \Cref{prop:folded_effective_degree_clean}: composing the degree-\(q\) perturbation-gradient surrogate with the degree-\(K_s\) sign polynomial and the degree-\(K_c\) clipping polynomial gives \(D_A\le qK_sK_c\). Each learner substep keeps the perturbation block unchanged and composes the current state with a degree-\(q\) gradient surrogate, so \(\deg U_{t,\ell}\le q\). Since degrees multiply under composition of polynomial maps, composing \(K_t\) attack substeps and then \(L_t\) learner substeps gives the first inequality in \eqref{eq:fixed_KL_degree_bound_clean}, and the uniform bounds \(K_t\le K_{\max}\), \(L_t\le L_{\max}\) give the displayed time-uniform cutoff \(D_{\mathrm{sched}}\). The coefficient expansion is the standard polynomial expansion on \(\mathbb R^d\), with zero padding above the true degree if needed.

\medskip\noindent\textbf{Item~(2).}
Let \(M_t:=K_t+L_t\), and enumerate the exact and polynomial substeps in application order by
\[
\Phi_{t,1}^{\mathrm{phys}},\dots,\Phi_{t,M_t}^{\mathrm{phys}},
\qquad
\Phi_{t,1},\dots,\Phi_{t,M_t},
\]
so that the first \(K_t\) maps are the attack substeps and the last \(L_t\) maps are the learner substeps. Starting from the same \(v\in\mathcal V\), define
\[
z_0=\tilde z_0:=v,
\qquad
z_r:=\Phi_{t,r}^{\mathrm{phys}}(z_{r-1}),
\qquad
\tilde z_r:=\Phi_{t,r}(\tilde z_{r-1}),
\qquad r=1,\dots,M_t.
\]
Forward invariance keeps every intermediate point inside \(\mathcal V\). By the triangle inequality, the uniform substep discrepancy bound, and the Lipschitz bound,
\[
\|z_r-\tilde z_r\|_2
\le
\|\Phi_{t,r}^{\mathrm{phys}}(z_{r-1})-\Phi_{t,r}(z_{r-1})\|_2
+
\|\Phi_{t,r}(z_{r-1})-\Phi_{t,r}(\tilde z_{r-1})\|_2
\le
\varepsilon_{\mathrm{sub}}+\Lambda\|z_{r-1}-\tilde z_{r-1}\|_2.
\]
Set \(e_r:=\|z_r-\tilde z_r\|_2\), so \(e_0=0\) and
\[
e_r\le \varepsilon_{\mathrm{sub}}+\Lambda e_{r-1},
\qquad r=1,\dots,M_t.
\]
An induction on \(r\) gives
\[
e_r\le \left(\sum_{q=0}^{r-1}\Lambda^q\right)\varepsilon_{\mathrm{sub}}.
\]
Therefore
\[
\|z_{M_t}-\tilde z_{M_t}\|_2
=
e_{M_t}
\le
\left(\sum_{q=0}^{M_t-1}\Lambda^q\right)\varepsilon_{\mathrm{sub}}
=
C_t(\Lambda)\,\varepsilon_{\mathrm{sub}}
\le
C_{\mathrm{sched}}(\Lambda)\,\varepsilon_{\mathrm{sub}}.
\]
Since \(z_{M_t}=\Psi_t^{\mathrm{phys},K_t,L_t}(v)\) and \(\tilde z_{M_t}=\Psi_t^{K_t,L_t}(v)\), this proves \eqref{eq:fixed_KL_base_step_bound_clean}.

\medskip\noindent\textbf{Item~(3).}
Once \(\Psi_t^{K_t,L_t}\) is viewed as the polynomial outer map for one training step, \Cref{thm:one_shot_folded_clean} uses only its coefficient family \(\{Q_\ell^{(t)}\}\), the associated truncated lifted blocks, and the hypotheses (H1)--(H4). Item~(1) provides the time-uniform degree bound needed to define those coefficients, and item~(2) gives \eqref{eq:base_step_bridge_assumption_clean} with the choice \eqref{eq:fixed_KL_eps_base_clean}. Therefore item~(2) of \Cref{prop:base_to_lifted_bridge_clean} yields the same lifted error-term bound, and Corollary~\ref{cor:exact_dynamics_bridge_clean} follows without further changes.
\end{proof}

A qRAM specialization follows under the following additional assumptions. Under the hypotheses of \Cref{thm:one_shot_folded_clean}, assume in addition that
\begin{equation}
\label{eq:qram_prep_cost_clean}
C_{\mathrm{prep}}(B_{\mathrm{rhs}})
=
\widetilde{\mathcal O}(\operatorname{polylog} N_h),
\end{equation}
for example because the entries of \(B_{\mathrm{rhs}}\) are stored in qRAM and the corresponding amplitude-loading routine is available.
Assume also that one query to the chosen QLSA access model for the rescaled horizon matrix \(\overline M\) can be implemented in
\begin{equation}
\label{eq:qram_sparse_query_cost_clean}
\widetilde{\mathcal O}(1)
\quad\text{or more generally}\quad
\widetilde{\mathcal O}(\operatorname{polylog} N_h)
\end{equation}
gates.
Then the gate complexity in item~(3) of \Cref{thm:one_shot_folded_clean} simplifies to
\begin{equation}
\label{eq:qram_one_shot_gate_complexity_clean}
\widetilde{\mathcal O}\!\left(
s_M\,\kappa_2(M)\,
\operatorname{polylog}\!\frac{N_h}{\varepsilon_{\mathrm{LS}}}
\right),
\end{equation}
while the success probability from item~(3) and the state-error bounds
\eqref{eq:one_shot_end_to_end_state_clean} and
\eqref{eq:one_shot_final_eps_clean} remain unchanged. Substitute \eqref{eq:qram_prep_cost_clean} and \eqref{eq:qram_sparse_query_cost_clean} into \eqref{eq:one_shot_gate_complexity_clean}; both contributions are then absorbed into the displayed \(\widetilde{\mathcal O}\) bound.

Segmented local contractivity can also be stated directly at the dynamics level. Partition the training window as
\[
0=T_0<T_1<\cdots<T_R=T,
\qquad
L_r:=T_{r+1}-T_r.
\]
For each segment $r\in\{0,1,\dots,R-1\}$, let
\[
Y_{N,r}^{\mathrm{ex}}
:=
\bigl(y^{(N)}(T_r),y^{(N)}(T_r+1),\dots,y^{(N)}(T_{r+1})\bigr)
\]
denote the exact truncated lifted trajectory restricted to that segment, and let \(\hat Y_r\) be the solution of the same truncated lifted recurrence on the same interval, initialized with the exact boundary value \(\hat y(T_r)=y^{(N)}(T_r)\). Assume that on each segment one has
\begin{equation}
\label{eq:segment_contractivity_clean}
\sup_{T_r\le t\le T_{r+1}-1}\|B(t)\|_2\le \rho_r<1,
\end{equation}
and define the segmentwise tail constant
\begin{equation}
\label{eq:segment_Gamma_clean}
\Gamma_{N,r}
:=
\max_{T_r\le t\le T_{r+1}-1}
\left(
\sum_{j=1}^{N}
\left[
\sum_{s=N+1}^{jD}\|K_{j,s}(t)\|_2\,\bar v_r^{s}
\right]^2
\right)^{1/2},
\end{equation}
where $\bar v_r<1$ bounds the exact trajectory of the polynomial update model on that segment. Then
\begin{equation}
\label{eq:segment_truncation_clean}
\|Y_{N,r}^{\mathrm{ex}}-\hat Y_r\|_2
\le
\sqrt{L_r+1}\,\frac{\Gamma_{N,r}}{1-\rho_r},
\end{equation}
and, if \(\hat Y^{\mathrm{seg}}\) denotes the concatenated segmented truncated lifted trajectory obtained by discarding duplicate terminal copies at interior boundaries, then
\begin{equation}
\label{eq:segment_global_truncation_clean}
\|Y_N^{\mathrm{ex}}-\hat Y^{\mathrm{seg}}\|_2^2
\le
\sum_{r=0}^{R-1}
(L_r+1)
\left(\frac{\Gamma_{N,r}}{1-\rho_r}\right)^2.
\end{equation}
The proof is the same contractive-forcing argument used in Item~(2) of \Cref{thm:truncation_cutoff_clean}, applied on each shorter window and then combined with the contractivity of the coordinate-deletion map that removes duplicated boundary blocks. This result controls only the dynamical approximation; restart and readout costs require separate analysis.

A final specialization worth recording is the affine-gradient case. Assume that on a local domain \(\mathcal V\subset \mathbb R^d\) the exact gradients are affine:
\begin{equation}
\label{eq:affine_gradient_class_clean}
g_{\delta,t}(v)=A_{\delta,t}v+b_{\delta,t},
\qquad
g_{u,t}(v)=A_{u,t}v+b_{u,t},
\qquad
t=0,1,\dots,T-1.
\end{equation}
Choose the polynomial gradient approximations to be exact, namely
\(
\mathcal G_{\delta,t}=g_{\delta,t}
\)
and
\(
\mathcal G_{u,t}=g_{u,t}
\).
Then \(q=1\), \(\varepsilon_{\delta,\mathrm{grad}}=0\), and \(\varepsilon_{u,\mathrm{grad}}=0\).
Assume further that:
\begin{enumerate}
\item the local domain \(\mathcal V\) is forward invariant under both the exact one-step map and the polynomial one-step model;
\item every \(v\in\mathcal V\) satisfies the dead-zone condition \eqref{eq:dead_zone_condition_clean} and the clip-safe condition \eqref{eq:clip_safe_region_clean};
\item there exists a target \(\varepsilon_{\mathrm{nl,step}}>0\) such that the sign and clip degrees are fixed once for the whole training window and chosen, for example by Item~(4) of \Cref{prop:sign_clip_design_clean} with the common local-domain parameters and \(\eta_\delta=\sup_{0\le t\le T-1}\eta_{\delta,t}\), so that
\begin{equation}
\label{eq:affine_uniform_nl_budget_clean}
\sup_{0\le t\le T-1}\ \sup_{v\in\mathcal V}
\|\delta_t^{\mathrm{ex}}(v)-\delta_t^{\mathrm{poly}}(v)\|_2
\le
\varepsilon_{\mathrm{nl,step}};
\end{equation}
\item \(\|v\|_2\le \bar v<1\) for every \(v\in\mathcal V\);
\item if \(A_{u,t}^{(\delta)}\in\mathbb R^{n\times m}\) denotes the perturbation block of \(A_{u,t}\), then \(\|A_{u,t}^{(\delta)}\|_2\le L_{u,\delta}^{(t)}\);
\item the full-step coefficients \(Q_{\ell}^{(t)}\) satisfy the weighted-tail condition \eqref{eq:weighted_coeff_uniform_clean} and the coefficient-sparsity condition used in Item~(1) of \Cref{prop:sparsity_oracle_clean}.
\end{enumerate}
Then the one-step approximation bound can be taken as
\begin{equation}
\label{eq:affine_base_step_budget_clean}
\varepsilon_{\mathrm{base,step}}
\le
\sup_{0\le t\le T-1}
\bigl(1+\eta_{u,t}\|A_{u,t}^{(\delta)}\|_2\bigr)\,
\varepsilon_{\mathrm{nl,step}}.
\end{equation}
Consequently, the modeling error bound in \eqref{eq:bridge_model_step_budget_clean} holds with this choice of \(\varepsilon_{\mathrm{base,step}}\), the cutoff \(N\) may be chosen by \eqref{eq:weighted_cutoff_design_clean}, and \(s_B\) may be bounded by \eqref{eq:sB_from_coeff_sparsity_clean} or \eqref{eq:sB_uniform_coeff_sparsity_clean}. If, in addition, the contractive and initialization assumptions (H1)--(H4) of Theorem~\ref{thm:one_shot_folded_clean} hold, then Corollary~\ref{cor:exact_dynamics_bridge_clean} yields a state-preparation guarantee for the exact nonsmooth dynamics. Since the exact gradients are affine, Item~(1) of \Cref{prop:base_to_lifted_bridge_clean} applies with \(\varepsilon_{u,\mathrm{grad}}=0\), and Item~(2) then gives the displayed modeling error bound.

\bibliographystyle{unsrtnat}
\bibliography{arxivbib}

\end{document}